\documentclass[11pt]{article}

\usepackage{amssymb,amsthm,amsmath,amssymb,wrapfig,dsfont,authblk}
\usepackage{thmtools,thm-restate}
\PassOptionsToPackage{hyphens}{url}
\usepackage[round,sort,semicolon]{natbib}
\usepackage{varioref}
\usepackage[usenames,svgnames,xcdraw,table]{xcolor}
\definecolor{DarkBlue}{rgb}{0.1,0.1,0.5}
\definecolor{DarkGreen}{rgb}{0.1,0.5,0.1}
\usepackage[backref=page]{hyperref}
\hypersetup{
     colorlinks   = true,
     linkcolor    = DarkBlue, 
     urlcolor     = DarkBlue, 
	 citecolor    = DarkGreen 
}
\renewcommand*{\backref}[1]{}
\renewcommand*{\backrefalt}[4]{%
    \ifcase #1 (Not cited.)%
    \or        (Cited on page~#2)%
    \else      (Cited on pages~#2)%
    \fi}
\usepackage[capitalise,noabbrev]{cleveref}
\usepackage[margin=0.9in]{geometry}

\usepackage[linesnumbered,lined,boxed,ruled,vlined]{algorithm2e} 

\SetKwInOut{Parameter}{Parameter}
\SetAlFnt{\small}
\SetAlCapFnt{\small}
\SetAlCapNameFnt{\small}
\SetAlCapHSkip{0pt}
\IncMargin{-\parindent}

\usepackage{tikz}
\usetikzlibrary{fit,calc}
\newcommand*{\tikzmk}[1]{\tikz[remember picture,overlay,] \node (#1) {};\ignorespaces}
\newcommand{\boxit}[1]{\tikz[remember picture,overlay]{\node[yshift=3pt,xshift=4pt,fill=#1,opacity=.25,fit={(A)($(B)+(1.0\linewidth,.8\baselineskip)$)}] {};}\ignorespaces}
\colorlet{pink}{red!40}
\colorlet{blue}{cyan!60}
\colorlet{mygray}{gray!55}

\makeatletter
\renewcommand{\paragraph}{%
  \@startsection{paragraph}{4}%
  {\z@}{1.0ex \@plus 1ex \@minus .2ex}{-1em}%
  {\normalfont\normalsize\bfseries}%
}
\let\oldnl\nl
\newcommand{\nonl}{\renewcommand{\nl}{\let\nl\oldnl}}
\makeatother

\usepackage{nicefrac}
\usepackage{enumerate}
\usepackage[labelfont={normalfont,bf},textfont=it]{caption}
\usepackage{subcaption}
\usepackage[T1]{fontenc}
\usepackage{mathtools}
\usepackage{bbm}

\usepackage{tabulary}
\usepackage{multirow}
\usepackage{booktabs}
\usepackage{colortbl}

\usepackage[backgroundcolor=blue!20!white]{todonotes}

\newtheorem{lemma}{Lemma}
\newtheorem{claim}{Claim}
\newtheorem{corollary}{Corollary}

\newtheorem{definition}{Definition}
\theoremstyle{definition}
\newtheorem{remark}[definition]{Remark}

\allowdisplaybreaks

\newcommand{\Alg}{\textrm{\textsc{Alg}}}
\newcommand{\APXhard}{\text{APX-hard}}
\newcommand{\BuildHierarchy}{\textrm{\textsc{BuildHierarchy}}}

\newcommand{\coNPc}{\text{co-NP-complete}}
\renewcommand{\d}{\mathbf d}
\newcommand{\e}{\mathbf e}
\newcommand{\EF}{\mathrm{EF}}
\newcommand{\EFone}{\textrm{EF1}}
\newcommand{\F}{{\mathcal F}}

\newcommand{\FNP}{\text{FNP}}
\newcommand{\fPO}{\textrm{fPO}}

\newcommand{\G}{{\mathcal G}}
\renewcommand{\H}{\mathcal H}

\newcommand{\I}{{\mathcal I}}
\newcommand{\level}{{\mathrm{level}}}
\newcommand{\MBB}{\mathrm{MBB}}

\newcommand{\N}{{\mathbbm N}}

\newcommand{\NPhard}{\text{NP-hard}}
\newcommand{\NP}{\text{NP}}

\newcommand{\NW}{\mathrm{NSW}}
\newcommand{\NSW}{\mathrm{NSW}}
\renewcommand{\O}{\mathcal O}
\newcommand{\p}{\mathbf p}
\renewcommand{\P}{\text{P}}

\newcommand{\pEF}{\mathrm{pEF}}
\newcommand{\pEFone}{\mathrm{pEF1}}
\newcommand{\PO}{\textrm{PO}}
\newcommand{\poly}{\mathrm{poly}}
\newcommand{\q}{\mathbf q}

\newcommand{\TFNP}{\mathrm{TFNP}}
\renewcommand{\u}{\mathbf u}
\newcommand{\V}{\mathcal V}

\newcommand{\x}{\mathbf x}
\newcommand{\opt}{\mathbf \omega}
\newcommand{\X}{\mathcal X}
\newcommand{\y}{\mathbf y}
\newcommand{\z}{\mathbf z}

\begin{document}

\title{\bfseries Finding Fair and Efficient Allocations}

\author{Siddharth Barman\thanks{Indian Institute of Science. \texttt{barman@csa.iisc.ernet.in} \\ \hspace*{13pt} Supported in part by a Ramanujan Fellowship (SERB - {SB/S2/RJN-128/2015}).}, \quad Sanath Kumar Krishnamurthy\thanks{Chennai Mathematical Institute. \texttt{sanathkumar9@cmi.ac.in}}, \quad Rohit Vaish\thanks{Indian Institute of Science. \texttt{rohitv@iisc.ac.in}}}

\date{}
\maketitle

\begin{abstract}
We study the problem of allocating a set of \emph{indivisible} goods among a set of agents in a \emph{fair} and \emph{efficient} manner. An allocation is said to be \emph{fair} if it is envy-free up to one good (EF1), which means that each agent prefers its own bundle over the bundle of any other agent up to the removal of one good. In addition, an allocation is deemed \emph{efficient} if it satisfies Pareto efficiency. While each of these well-studied properties is easy to achieve separately, achieving them together is far from obvious. Recently, \citet{CKM+16unreasonable} established the surprising result that when agents have additive valuations for the goods, there always exists an allocation that simultaneously satisfies these two seemingly incompatible properties. Specifically, they showed that an allocation that maximizes the Nash social welfare objective is both EF1 and Pareto efficient. However, the problem of maximizing Nash social welfare is $\NPhard{}$. As a result, this approach does not provide an efficient algorithm for finding a fair and efficient allocation.

In this paper, we bypass this barrier, and develop a pseudopolynomial time algorithm for finding allocations that are EF1 and Pareto efficient; in particular, when the valuations are bounded, our algorithm finds such an allocation in polynomial time. Furthermore, we establish a stronger existence result compared to \citet{CKM+16unreasonable}: For additive valuations, there always exists an allocation that is EF1 and \emph{fractionally} Pareto efficient.

Another key contribution of our work is to show that our algorithm provides a polynomial-time 1.45-approximation to the Nash social welfare objective. This improves upon the best known approximation ratio for this problem (namely, the 2-approximation algorithm of \citealp{CDG+17convex}), and also matches the lower bound on the integrality gap of the convex program of \citet{CDG+17convex}. Unlike many of the existing approaches, our algorithm is completely combinatorial, and relies on constructing \emph{integral} Fisher markets wherein specific equilibria are not only efficient, but also fair.
\end{abstract}

\section{Introduction}
\label{sec:Introduction}
The theory of fair division addresses the fundamental problem of allocating goods or resources among agents in a fair and efficient manner. Such problems arise in many real-world settings such as government auctions, divorce settlements, and border disputes. Starting with the work of \citet{S48problem}, there is now a vast literature in economics and mathematics to formally address fair division \citep{BT96fair,M04fair,BCE+16handbook}. 
 Many interesting connections have also been found between fair division and fields such as topology, measure theory, combinatorics, and algorithms \citep{M08using}.

Much of the prior work in fair division, though, has focused on \emph{divisible} goods, which model resources that can be fractionally allocated (such as land). A standard fairness concept in this setting is \emph{envy-freeness}~\citep{F67resource}, which requires that each agent prefers its own allocation over that of any other agent. A well-known result of \citet{V74equity} shows that
for the divisible setting, there always exists an allocation that is both envy-free (i.e., fair) and \emph{Pareto efficient}. Furthermore, such an allocation can be computed in polynomial time~\citep{EG59consensus,DPS+08market}. These results, however, do not extend to the setting of \emph{indivisible} goods, which model discrete resources such as courses at universities~\citep{OSB10finding} or inherited artwork. In fact, many of the classical solution concepts and algorithms that have been developed for divisible goods are not directly applicable to the indivisible setting. For example, an envy-free allocation fails to exist even in the simple setting of a single indivisible good and two agents.

These considerations have motivated recent work in the theoretical computer science and economics communities on developing relevant notions of fairness, along with existence results and algorithms for the problem of fairly allocating indivisible goods \citep{LMM+04approximately,B11combinatorial,BL16characterizing,KPW18fair}. We contribute to this line of work by showing that guarantees analogous to the fundamental result of \citet{V74equity} hold even for indivisible goods in terms of a natural and necessary relaxation of envy-freeness. Specifically, we show that for additive valuations,\footnote{\emph{Additivity} means that an agent's valuation for a set of goods is the sum of its valuations for the individual goods in that set.} a fair and efficient allocation always exists, and such an allocation can be computed in (pseudo)-polynomial time.

We consider an allocation of indivisible goods to be fair if it is \emph{envy-free up to one good} ($\EFone{}$). This notion was defined by \citet{B11combinatorial}, and provides a compelling relaxation of the envy-freeness property.\footnote{The notion of \EFone{} has found practical appeal on the popular fair division website ``Spliddit'' \citep{GP15spliddit} and in course allocation at Wharton School of Business \citep{BCK+16course}.} An allocation is said to be $\EFone{}$ if each agent prefers its own bundle over the bundle of any other agent up to the removal of the most valuable good from the other agent's bundle. Although the existence of envy-free allocations is not guaranteed in the context of indivisible goods, an $\EFone{}$ allocation always exists---even under general, combinatorial valuations---and can be found in polynomial time~\citep{LMM+04approximately}.

With this notion of fairness in hand, it is relevant to ask whether we can achieve \emph{efficiency along with fairness} while allocating indivisible goods.\footnote{Note that fairness, by itself, does not guarantee efficiency; in fact, an $\EFone{}$ allocation can be highly inefficient (\Cref{subsec:EF1-inefficient}).} This question was recently studied by \citet{CKM+16unreasonable}, who showed a striking result that there is no need to trade efficiency for fairness: For additive valuations, an allocation that maximizes the \emph{Nash social welfare}~\citep{N50bargaining,KN79nash}---defined to be the geometric mean of the agents' valuations---is both fair ($\EFone{}$) and Pareto efficient. However, maximizing the Nash social welfare ($\NSW$) over integral allocations is an \NPhard{} problem \citep{NNR+14computational}. (In fact, the problem is known to be \APXhard{} \citep{Lee17APX}). Therefore, this existence result does not automatically provide an efficient algorithm for finding a fair and efficient allocation of indivisible goods. Our work bypasses this limitation by providing \emph{a pseudopolynomial time algorithm for finding an $\EFone{}$ and Pareto efficient allocation} of indivisible goods under additive valuations. In particular, when the valuations are \emph{bounded}, our algorithm finds such an allocation in \emph{polynomial time}. It is worth pointing out that the problem of maximizing $\NSW$ remains \APXhard{} even for bounded valuations \citep{Lee17APX}.

A related problem is that of developing approximation algorithms for $\NSW$ maximization.
This problem has received considerable attention in recent years~\citep{CG15approximating,AGM+18nash,AGS+17nash,BGH+17earning,CDG+17convex,GHM18Approximating}. The first constant-factor (specifically, $2.89$) approximation for this problem was provided by \citet{CG15approximating}. This approximation factor was subsequently improved to $e$ \citep{AGS+17nash}, and most recently to $2$ \citep{CDG+17convex}. Similar approximation guarantees have also been developed for more general market models such as piecewise-linear concave utilities \citep{AGM+18nash}, budget additive valuations \citep{GHM18Approximating}, and multi-unit markets \citep{BGH+17earning}. 

While the problem of approximating $\NSW$ is interesting in its own right, it is relevant to note that an allocation that approximates this objective is, in and of itself, not guaranteed to be EF1 or Pareto efficient (see \Cref{subsec:ApproxNash_NotEF1_NotPO} for an example).\footnote{We also provide an example (\Cref{subsec:spending-restricted}) in which \emph{every} rounding of the ``spending restricted outcome''---a market equilibrium notion used in the design of approximation algorithms for NSW~\citep{CG15approximating,AGM+18nash,CDG+17convex}---violates EF1.} A second key contribution of our work is to show that our algorithm provides \emph{a polynomial-time $1.45$-approximation to the $\NSW$ maximization problem}. Thus, not only does our algorithm improve upon the best-known approximation ratio for this problem (namely, the $2$-approximation algorithm of \citet{CDG+17convex}), it is also guaranteed to return a fair and efficient outcome. The following list summarizes our contributions.

\paragraph{Our contributions}
\begin{itemize}
	\item We develop an \emph{algorithm} for computing an $\EFone{}$ and Pareto efficient allocation for additive valuations. The running time of our algorithm is pseudopolynomial for general integral valuations (\Cref{THM:EF1+PO_BOUNDEDVALUATIONS_PSEUDOPOLYTIME}) and polynomial when the valuations are bounded (\Cref{rem:EF1+PO_BoundedValuations_PolyTime}). In addition, our algorithm can find an approximate $\EFone$ and approximate Pareto efficient allocation in polynomial time even without the bounded valuations assumption (\Cref{rem:EpsEFoneDeltaPOPolytime}).
	
	\item We establish a \emph{stronger existence result} compared to \citet{CKM+16unreasonable}: For additive valuations, there always exists an allocation that is $\EFone{}$ and \emph{fractionally} Pareto efficient~(\Cref{thm:EF1+fPO_Existence}). In other words, the problem of finding an $\EFone{}$ and fractionally Pareto efficient allocation is \emph{total}. An interesting complexity-theoretic implication of this result is that there exists a nondeterministic polynomial time algorithm for finding an $\EFone{}$ and Pareto efficient allocation (\Cref{rem:EF+PO_NonDeterministicPolytime}). This implication does not directly follow from the existence result of \citet{CKM+16unreasonable}, as the problem of verifying whether an arbitrary allocation is Pareto efficient is known to be \coNPc{}~\citep{KBK+09complexity}.
	
	\item We show that our algorithm provides a \emph{polynomial-time $1.45$-approximation} for the Nash social welfare ($\NSW$) maximization problem (\Cref{THM:APPROXNASH}). This improves upon the best known approximation factor for this problem (namely, the 2-approximation algorithm of \citet{CDG+17convex}), and also matches the lower bound of $e^{1/e} \approx 1.44$ on the integrality gap of the convex program of \citet{CDG+17convex}. 
	An interesting byproduct of our analysis is a novel connection between envy-freeness and $\NSW$: Under identical valuations, an $\EFone$ allocation provides a $1.45$-approximation to the maximum $\NSW$ (\Cref{lemma:id-val}).
\end{itemize}

\paragraph{Our techniques}  
It is known from the fundamental theorems of welfare economics that markets tend toward efficiency. Intuitively, our results are based on establishing a complementary result that markets can be \emph{fair} as well. In particular, we construct a Fisher market along with an underlying equilibrium which is integral (i.e., corresponding to an allocation of the indivisible goods) and $\EFone{}$. The fact that this allocation is a market equilibrium ensures, via the first welfare theorem, that it is Pareto efficient as well. 

More concretely, we start with a Pareto efficient allocation, and iteratively modify the allocation by exchanging goods between the agents. The goal of the exchange  step is to locally move toward a fair allocation. Additionally, throughout these exchanges, we maintain a set of prices that ensure that the current allocation corresponds to an equilibrium outcome for the existing market. We stop when the equilibrium of the market (i.e., the allocation at hand) satisfies \emph{price envy-freeness up to one good} (refer to \Cref{subsec:MarketTerminology} for a formal definition). Essentially, this property ensures that under the given market prices, the spending of an agent is at least that of any other agent up to the removal of the highest priced good from the other agent's bundle. Requiring the spendings to be balanced in this manner implies the desired $\EFone{}$ property for the corresponding fair division instance; see \Cref{sec:Algorithm} for a detailed description of this construction.

At a conceptual level, our approach differs from the existing approaches in two important ways: First, our algorithm works with an \emph{integral} Fisher market at every step, thereby breaking away from the standard \emph{relax-and-round} paradigm where a  fractional market equilibrium is first computed (typically as a solution of some convex program) followed by a rounding step \citep{CG15approximating,CDG+17convex,AGS+17nash,AGM+18nash,GHM18Approximating}. Second, unlike all existing approaches, our algorithm uses the notion of \emph{price envy-freeness up to one good} as a measure of balanced spending in the Fisher market. To the best of our knowledge, this notion is novel to this work, and might find future use in the design of fair and efficient algorithms for other settings.

\section{Preliminaries}
\label{sec:Preliminaries}

\subsection{The Fair Division Model}
\label{subsec:Model}

\paragraph{Problem instance} 
An \emph{instance} of the fair division problem is a tuple $\langle [n], [m], \V \rangle$, where $[n] = \{1,2,\dots,n\}$ denotes the set of $n \in \mathbb{N}$ \emph{agents}, $[m] = \{1,2,\dots,m\}$ denotes the set of $m \in \mathbb{N}$ \emph{goods}, and the \emph{valuation profile} $\V = \{v_1,v_2,\dots,v_n\}$ specifies the preferences of each agent $i \in [n]$ over the set of goods $[m]$ via a \emph{valuation function} $v_i: 2^{[m]} \rightarrow \mathbb{Z_+} \cup \{0\}$. We will assume throughout that the valuation functions are \emph{additive}, i.e., for each agent $i \in [n]$ and any set of goods $G \subseteq [m]$, $v_i(G) := \sum_{j \in G} v_i(\{j\})$.\footnote{We will assume that $v_i(\{\emptyset\}) = 0$ for all $i \in [n]$.} For simplicity, we will write $v_{i,j}$ instead of $v_i(\{j\})$ for a singleton good $j \in [m]$. Thus, $v_{i,j}$ is non-negative and integral for each agent $i \in [n]$ and each good $j \in [m]$. We will also assume, without loss of generality, that for each good $j \in [m]$, there exists some agent $i \in [n]$ with a nonzero valuation for it, i.e., $v_{i,j} > 0$. Finally, we let $v_{\max} \coloneqq \max_{i,j} v_{i,j}$.

\paragraph{Allocation}
An \emph{allocation} $\x \in \{0,1\}^{n \times m}$ refers to an $n$-partition $(\x_1,\dots,\x_n)$ of $[m]$, where $\x_i \subseteq [m]$ is the \emph{bundle} allocated to agent $i$. We let $\X$ denote the set of all $n$ partitions of $[m]$. Given an allocation $\x$, the valuation of an agent $i \in [n]$ for the bundle $\x_i$ is $v_i(\x_i) = \sum_{j \in \x_i} v_{i,j}$.

Another useful notion is that of a fractional allocation. A \emph{fractional allocation} $\x \in [0,1]^{n \times m}$ refers to a (possibly) fractional assignment of the goods to the agents such that no more than one unit of each good is allocated, i.e., for all $j \in [m]$, we have $\sum_{i \in [n]} x_{i,j} \leq 1$. We will use the term \emph{allocation} to refer to an integral allocation, and explicitly write \emph{fractional allocation} otherwise.

\subsection{Fairness Notions}
\label{subsec:FairnessNotions}

\paragraph{Envy-freeness and its variants}
Given an instance $\langle [n], [m], \V \rangle$ and an allocation $\x$, we say that an agent $i \in [n]$ \emph{envies} another agent $k \in [n]$ if $i$ strictly prefers the bundle of $k$ over its own bundle, i.e., $v_i(\x_k) > v_i(\x_i)$. An allocation $\x$ is said to be \emph{envy-free} ($\EF{}$) if each agent prefers its own bundle over that of any other agent, i.e., for every pair of agents $i,k \in [n]$, we have $v_i(\x_i) \geq v_i(\x_k)$.

An allocation $\x$ is said to be \emph{envy-free up to one good} (\EFone{}) if for every pair of agents $i,k \in [n]$, there exists a good $j \in \x_k$ such that $v_i(\x_i) \geq v_i(\x_k \setminus \{j\})$. Given any $\varepsilon > 0$, an allocation $\x$ is said to be \emph{$\varepsilon$-approximately envy-free up to one good} ($\varepsilon$-\EFone{}) if for every pair of agents $i,k \in [n]$, there exists a good $j \in \x_k$ such that $(1+\varepsilon) v_i(\x_i) \geq v_i (\x_k \setminus \{j\})$.

\paragraph{Nash social welfare}
Given an allocation $\x$, write $\NW(\x) := \left( \prod_{i \in [n]} v_i(\x_i) \right)^{\frac{1}{n}}$ to denote the \emph{Nash social welfare} of $\x$. An allocation $\x^*$ said to be \emph{Nash optimal} if $\textstyle{ \x^* \in \arg\max_{\x \in \X} \NW(\x) }$.

\subsection{Efficiency Notions}
\label{subsec:EfficiencyNotions}

\paragraph{Pareto efficiency}
Given an instance $\langle [n], [m], \V \rangle$ and an allocation $\x$, we say that $\x$ is Pareto dominated by another allocation $\y$ if $v_k(\y_k) \geq v_k(\x_k)$ for every agent $k \in [n]$, and  $v_i(\y_i) > v_i(\x_i)$ for some agent $i \in [n]$. An allocation is said to be \emph{Pareto efficient} or \emph{Pareto optimal} (\PO{}) if it is not Pareto dominated by any other allocation. 
Similarly, $\x$ is \emph{$\varepsilon$-Pareto efficient} ($\varepsilon$-$\PO$) if it is not $\varepsilon$-Pareto dominated by any other allocation $\y$, i.e., there does not exist an allocation $\y$ such that $v_k(\y_k) \geq (1+\varepsilon) v_k(\x_k)$ for every agent $k \in [n]$ and $v_i(\y_i) > (1+\varepsilon) v_i(\x_i)$ for some agent $i \in [n]$.

Some of our results use a generalization of Pareto efficiency, which we call \emph{fractional Pareto efficiency}. An allocation is said to be fractionally Pareto efficient (\fPO{}) if it not Pareto dominated by any fractional allocation. Thus, a fractionally Pareto efficient allocation is also Pareto efficient, but the converse is not necessarily true (\Cref{subsec:approxNSW+fPO_non-existance} provides an example).

\section{Main Results}
\label{sec:Results_Highlights}

This section provides the statements of our three main results: an algorithm for finding an $\EFone$ and $\PO$ allocation (\Cref{THM:EF1+PO_BOUNDEDVALUATIONS_PSEUDOPOLYTIME}), an existence result for $\EFone$ and $\fPO$ allocation (\Cref{thm:EF1+fPO_Existence}), and an approximation algorithm for Nash social welfare (\Cref{THM:APPROXNASH}).

\paragraph{Algorithmic Result:}

\begin{restatable}{theorem}{EFonePOBoundedValuationsPseudoPolyTime}
 \label{THM:EF1+PO_BOUNDEDVALUATIONS_PSEUDOPOLYTIME}
 Given any fair division instance $\I = \langle [n], [m], \V \rangle$ with additive valuations, an allocation that is envy-free up to one good $(\EFone{})$ and Pareto efficient $(\PO{})$ can be found in $\O\left( \poly(m,n, v_{\max} ) \right)$ time, where $v_{\max} = \max_{i,j} v_{i,j}$.
\end{restatable}

\begin{remark}
\label{rem:EF1+PO_BoundedValuations_PolyTime}
Note that when all valuations are \emph{polynomially bounded} (i.e., there exists a polynomial $f(m,n)$ such that for all $i \in [n]$ and $j \in [m]$, $v_{i,j} \leq f(m,n)$), an $\EFone$ and $\PO$ allocation can be computed in \emph{polynomial} time. In particular, this is true when all valuations are bounded by a constant. As mentioned earlier in \Cref{sec:Introduction}, the problem of maximizing $\NSW$ remains \APXhard{} even for constant valuations \citep{Lee17APX}, and therefore our result circumvents the intractability associated with computing a Nash optimal allocation in order to achieve these two properties.
\end{remark}

\begin{remark}
\label{rem:EpsEFoneDeltaPOPolytime}
If we relax the fairness and efficiency requirements in \Cref{THM:EF1+PO_BOUNDEDVALUATIONS_PSEUDOPOLYTIME} to their approximate analogues, then our algorithm is guaranteed to run in polynomial time. Specifically, our algorithm can find an $\varepsilon$-$\EFone{}$ and $\varepsilon$-$\PO$ allocation in $\O\left( \poly(m,n, \frac{1}{\varepsilon}, \ln v_{\max} ) \right)$ time, where $v_{\max} = \max_{i,j} v_{i,j}$. 
(Refer to \Cref{lem:Approx_EF_Approx_PO_Polytime} in \Cref{subsec:ALG_analysis_generalcase}).
\end{remark}

The proof of \Cref{THM:EF1+PO_BOUNDEDVALUATIONS_PSEUDOPOLYTIME} is provided in \Cref{sec:MainResult_EFonePO_Pseudopoly}.

\paragraph{Existence Result:}

\begin{restatable}{theorem}{EFonefPOExistence}
 \label{thm:EF1+fPO_Existence}
 Given any fair division instance with additive valuations, there always exists an allocation that is envy-free up to one good $(\EFone{})$ and fractionally Pareto efficient $(\fPO{})$.
\end{restatable}

\begin{remark}
\label{rem:EF+PO_NonDeterministicPolytime}
Consider the canonical binary relation $\mathcal{R}^{\EFone{}+\PO{}}$ associated with the problem of finding an $\EFone{}$ and $\PO$ allocation, defined as follows: For a fair division instance $\mathcal{I}$ and an allocation $\x$, the relation $\mathcal{R}^{\EFone{}+\PO{}}(\I,\x)$ holds if and only if $\x$ is an $\EFone{}$ and $\PO{}$ allocation of $\I$. It is relevant to note that under standard complexity theoretic assumptions, $\mathcal{R}^{\EFone{}+\PO{}}$ is not in $\TFNP$.\footnote{It is known that determining whether an arbitrary allocation is $\PO$ is $\coNPc$ \citep{KBK+09complexity}. This fact can be used to show that verifying whether a given allocation is $\EFone$ and $\PO$ is also $\coNPc$. Hence, the binary relation $\mathcal{R}^{\EFone{}+\PO{}}$ cannot be efficiently verified (i.e., it is not in $\FNP$), unless $\P$ = $\NP$.} By contrast, the binary relation $\mathcal{R}^{\EFone{}+\fPO{}}(\I,\x)$, which holds if and only if $\x$ is an $\EFone$ and $\fPO$ allocation for the instance $\I$, admits efficient verification.\footnote{$\EFone$ can be checked by considering $\O(n^2)$ inequalities, and $\fPO$ can be verified by a linear program (\Cref{subsec:Second_Welfare_Theorem_for_Fisher_Markets}).} Since \Cref{thm:EF1+fPO_Existence} shows that $\mathcal{R}^{\EFone{}+\fPO{}}$ is total, we get that the binary relation $\mathcal{R}^{\EFone{}+\fPO{}}$ is in $\TFNP$. Thus, there exists a nondeterministic polynomial time algorithm for finding an $\EFone{}$ and $\fPO$ (and hence $\EFone{}$ and $\PO$) allocation.	
\end{remark}

The proof of \Cref{thm:EF1+fPO_Existence} is provided in \Cref{sec:EF1+fPO_Existence}.

\paragraph{Approximating Nash Social Welfare:}

\begin{restatable}{theorem}{ApproxNash}
 \label{THM:APPROXNASH}
 For additive valuations, there exists a polynomial-time $1.45$-approximation algorithm for the Nash social welfare maximization problem.
\end{restatable}

Our proof of \Cref{THM:APPROXNASH} draws on the following interesting connection between approximate envy-freeness and Nash social welfare:

\begin{restatable}{lemma}{EFoneNSWIdenticalVals}
\label{lemma:id-val}
Given a fair division instance with identical and additive valuations, any $\varepsilon\text{-}\EFone{}$ allocation provides a $e^{(1+\varepsilon)/e}$-approximation to Nash social welfare.
\end{restatable}

The proof of \Cref{THM:APPROXNASH} is provided in \Cref{sec:ApproxNash}.

\section{Our Algorithm}
\label{sec:Algorithm}


This section presents our algorithm. We start with the relevant preliminaries in \Cref{subsec:MarketTerminology} that provide the necessary definitions required for describing the algorithm. The pseudocode of the algorithm appears in \Cref{subsec:Description_of_ALG} along with a brief description.

\subsection{Market Terminology}
\label{subsec:MarketTerminology}

\paragraph{Fisher market} The Fisher market is a fundamental model in the economics of resource allocation \citep{BS00compute}. It captures the setting where a set of buyers enter the market with prespecified budgets, and use it to buy goods that provide maximum utility per unit of money spent. Specifically, a Fisher market consists of a set $[n] = \{1,2,\dots,n\}$ of $n$ \emph{buyers}, a set $[m] = \{1,2,\dots,m\}$ of $m$ divisible \emph{goods} (exactly one unit of each good is available), and a \emph{valuation profile} $\V = \{v_1,v_2,\dots,v_n\}$. Each buyer $i \in [n]$ has an initial \emph{endowment} (or budget) $e_i > 0$. The endowment holds no intrinsic value for a buyer and is only used for buying the goods. We call $\e = (e_1,\dots,e_n)$ the \emph{endowment vector}, and denote a \emph{market instance} by $\langle [n], [m], \V, \e \rangle$.

A \emph{market outcome} is given by the pair $\langle \x, \p \rangle$, where the \emph{allocation vector} $\x = (\x_1,\dots,\x_n)$ is a fractional allocation of the $m$ goods, and the \emph{price vector} $\p = (p_1,\dots,p_m)$ associates a price $p_j \geq 0$ with each good $j \in [m]$. The \emph{spending} of buyer $i$ under the market outcome $\langle \x, \p \rangle$ is given by $\p(\x_i) = \sum_{j = 1}^m x_{i,j} p_j$. The \emph{valuation} derived by the buyer $i$ under the market outcome $\langle \x, \p \rangle$ is given by $v_i(\x_i) =  \sum_{j=1}^m x_{i,j} v_{i,j}$.

Given a price vector $\p = (p_1,\dots,p_m)$, define the \emph{bang per buck} ratio of buyer $i$ for good $j$ as $\alpha_{i,j} \coloneqq v_{i,j}/p_j$, and its \emph{maximum bang per buck} ratio as $\alpha_i \coloneqq \max_j \alpha_{i,j}$.\footnote{If $v_{i,j} = 0$ and $p_j = 0$, then we define $\alpha_{i,j} = 0$.} Let $\MBB_i \coloneqq \{j \in [m] : v_{i,j}/p_j = \alpha_i\}$ denote the set of all goods that maximize the bang per buck ratio for buyer $i$ at the price vector $\p$. We call $\MBB_i$ the \emph{maximum bang per buck set} (or $\MBB{}$ set) of buyer $i$ at the price vector $\p$.

An outcome $\langle \x, \p \rangle$ is said to be a \emph{Fisher market equilibrium} if it satisfies the following conditions:
\begin{itemize}
	\item \emph{Market clearing}: Each good is either priced at zero or is completely allocated. That is, for each good $j \in [m]$, either $p_j = 0$ or $\sum_{i=1}^n x_{i,j} = 1$.
	\item \emph{Budget exhaustion}: Buyers spend their endowments completely, i.e., $\p(\x_i) = e_i$ for all $i \in [n]$.
	\item \emph{Maximum bang per buck allocation}: Each buyer's allocation is a subset of its $\MBB{}$ set. That is, for any buyer $i \in [n]$ and any good $j \in [m]$, $x_{i,j} > 0 \implies j \in \MBB_i$.
	Stated differently, each buyer only spends on its maximum bang per buck goods. Notice that a consequence of spending only on $\MBB{}$ goods is that each buyer maximizes its utility at the given prices $\p$ under the budget constaints.\footnote{One might expect utility maximization condition to be equivalent to the $\MBB{}$ allocation condition. \Cref{subsec:Fisher_counterexample} shows that this is not the case; indeed, $\MBB{}$ allocation is a strictly stronger requirement.}
\end{itemize}

We refer the reader to \Cref{sec:Market_Appendix} for additional market preliminaries.

\begin{restatable}[First Welfare Theorem; {\citealp[Chapter~16]{MWG+95microeconomic}}]{prop}{FirstWelfareTheorem}
 \label{prop:FirstWelfareTheorem}
 For a Fisher market with additive valuations, any equilibrium outcome is fractionally Pareto efficient ($\fPO$).
\end{restatable}

\paragraph{Price envy-freeness and its variants} Several of our results rely on constructing market outcomes with a property called \emph{price envy-freeness}---a notion we consider to be of independent interest. Specifically, let $\x$ be an allocation and let $\p$ be a price vector for a given Fisher market. We say that $\x$ is \emph{price envy-free} ($\pEF{}$) with respect to $\p$ if for every pair of buyers $i,k \in [n]$, we have $\p(\x_i) \geq \p(\x_k)$.\footnote{Equivalently, for every pair of buyers $i,k \in [n]$, we require $\p(\x_i) = \p(\x_k)$.} Similarly, $\x$ is said to be \emph{price envy-free up to one good} ($\pEFone{}$) with respect to $\p$ if for every pair of buyers $i,k \in [n]$, there exists a good $j \in \x_k$ such that $\p(\x_i) \geq \p \left( \x_k \setminus \{j\} \right)$. Finally, given any $\varepsilon > 0$, we say that an allocation $\x$ is \emph{$\varepsilon$-approximately price envy-free up to one good} ($\varepsilon$-$\pEFone{}$) with respect to $\p$ if for every pair of buyers $i,k \in [n]$, there exists a good $j \in \x_k$ such that $(1+\varepsilon) \p(\x_i) \geq \p (\x_k \setminus \{j\})$. 

\paragraph{$\MBB{}$ graph and alternating paths} The \emph{$\MBB{}$ graph} of a Fisher market instance with a price vector $\p$ is defined as a bipartite graph $G$ whose vertex set consists of the set of agents $[n]$ and the set of goods $[m]$, and there is an edge between an agent $i \in [n]$ and a good $j \in [m]$ if $j \in \MBB_i$ (called an \emph{$\MBB$ edge}). Given an allocation $\x$, we can augment the $\MBB{}$ graph by adding \emph{allocation edges}, i.e., an edge between an agent $i \in [n]$ and a good $j \in [m]$ such that $j \in \x_i$. For an augmented $\MBB$ graph, we define an \emph{alternating path} $P = (i,j_1,i_1,j_2,i_2,\dots,i_{\ell - 1},j_\ell,k)$ from agent $i$ to agent $k$ (and involving the agents $i_1,i_2,\dots,i_{\ell-1}$ and the goods $j_1,j_2,\dots,j_\ell$) as a series of alternating $\MBB{}$ and allocation edges such that $j_1 \in \MBB_i \cap \x_{i_1}$, $j_2 \in \MBB_{i_1} \cap \x_{i_2}$,$\dots$, $j_\ell \in \MBB_{i_{\ell - 1}} \cap \x_k$. If such a path exists, we say that the agent $k$ is \emph{reachable} from agent $i$ via an alternating path. Notice that no agent or good is allowed to repeat in an alternating path. We say that the path $P$ is of \emph{length} $2\ell$ since it consists of $\ell$ $\MBB{}$ edges and $\ell$ allocation edges.

\paragraph{Hierarchy structure} Let $G$ denote the augmented $\MBB{}$ graph for a Fisher market instance with the market outcome $(\x,\p)$. Fix a \emph{source} agent $i \in [n]$ in $G$. Define the \emph{level} of an agent $k \in [n]$ as half the length of the shortest alternating path from $i$ to $k$ (if one exists). The level of the source agent $i$ is defined to be zero. If there is no alternating path from $i$ to some agent $k$ in $G$ (i.e., if $k$ is not reachable from $i$), then the level of $k$ is set to be $n$. The \emph{hierarchy} structure $\H_i$ of agent $i$ is defined as a level-wise collection of all agents that are reachable from $i$, i.e., $\H_i = \{\H_i^{0},\H_i^{1},\H_i^{2},\dots,\}$, where $\H_i^{\ell}$ denotes the set of agents that are at level $\ell$ with respect to the agent $i$. \Cref{subsec:BuildHierarchy} provides a polynomial time subroutine called $\BuildHierarchy{}$ for constructing the hierarchy.

Given a hierarchy $\H_i$, we will overload the term \emph{alternating path} to refer to a series of alternating $\MBB{}$ and allocation edges \emph{connecting agents at a lower level to those at a higher level}. That is, a path $P = (i,j_1,i_1,j_2,i_2,\dots,i_{\ell - 1},j_\ell,k)$ involving agents from the hierarchy $\H_i$ is said to be an alternating path if (1) $j_1 \in \MBB_i \cap \x_{i_1}$, $j_2 \in \MBB_{i_1} \cap \x_{i_2}$,$\dots$, $j_\ell \in \MBB_{i_{\ell - 1}} \cap \x_k$, and (2) $\level(i) < \level(i_1) < \level(i_2) < \dots < \level(i_{\ell - 1}) < \level(k)$. In particular, an alternating path in a hierarchy cannot have edges between agents at the same level.

\paragraph{Violators and path-violators} Given a Fisher market instance and a market outcome $(\x,\p)$, an agent $i \in [n]$ with the smallest spending among all the agents is called the \emph{least spender}, i.e., $i \in \arg\min_{k \in [n]} \p(\x_k)$ (ties are broken according to a prespecified ordering over the agents). An agent $k \in [n]$ is said to be a \emph{violator} if for every good $j \in \x_k$, we have that $\p(\x_k \setminus \{j\}) > \p(\x_i)$, where $i$ is the least spender. Similarly, agent $k \in [n]$ is said to be an $\varepsilon$-\emph{violator} if for every good $j \in \x_k$, we have that $\p(\x_k \setminus \{j\}) > (1+\varepsilon) \p(\x_i)$. Notice that an agent can be a violator without being an $\varepsilon$-\emph{violator}. Also notice that if no agent is a violator ($\varepsilon$-violator), then the allocation $\x$ is $\pEFone{}$ ($\varepsilon$-$\pEFone{}$) with respect to $\p$.

A closely related notion is that of a \emph{path-violator}. Let $i$ denote the least spender, and let $\H_i$ denote the hierarchy of agent $i$. An agent $k \in \H_i$ is said to be a \emph{path-violator} with respect to the alternating path $P = (i,j_1,i_1,j_2,i_2,\dots,i_{\ell - 1},j_\ell,k)$ if $\p(\x_k \setminus \{j_\ell\}) > \p(\x_i)$. Observe that a path-violator (along a path $P$) need not be a violator, since there can be a good $j \in \x_k$ not on the path $P$ such that $\p(\x_k \setminus \{j\}) \leq \p(\x_i)$. Similarly, an agent $k \in \H_i$ is said to be an \emph{$\varepsilon$-path-violator} with respect to the alternating path $P = (i,j_1,i_1,j_2,i_2,\dots,i_{\ell - 1},j_\ell,k)$ if $\p(\x_k \setminus \{j_\ell\}) > (1+\varepsilon) \p(\x_i)$.

\subsection{Description of the Algorithm}
\label{subsec:Description_of_ALG}

Given any fair division instance $\I = \langle [n],[m],\V \rangle$ as input and a parameter $\varepsilon > 0$, our algorithm (Algorithm~\ref{alg:Main}), referred to as \Alg{} from here onwards, constructs a market equilibrium $(\x,\p)$ with respect to a Fisher market instance $\langle [n],[m],\V, \e \rangle$ (for a suitable endowment vector $\e$). The pair $(\x,\p)$ has the following two properties: (1) $\x$ is an integral allocation, and (2) $\x$ is $3\varepsilon$-$\pEFone{}$ with respect to $\p$. The second property allows us to show that the allocation $\x$ is $3\varepsilon$-\EFone{} for the corresponding fair division instance $\I$ (see \Cref{lem:eps-pEF1_implies_eps-EF1} below). Furthermore, by the first welfare theorem (\Cref{prop:FirstWelfareTheorem}), the allocation $\x$ is also guaranteed to be fractionally Pareto efficient (\fPO{}) for the Fisher market instance, and consequently for the fair division instance $\I$.

\begin{lemma}
 \label{lem:eps-pEF1_implies_eps-EF1}
 Let $\varepsilon \geq 0$, and let $\x$ and $\p$ be an allocation and a price vector respectively for a market instance $\langle [n], [m], \V, \e \rangle$ such that (1) $\x$ is $\varepsilon$-approximately price-envy-free up to one good $(\varepsilon\text{-}\pEFone{})$, and (2) $\x_i \subseteq \MBB_i$ for each buyer $i \in [n]$. Then, $\x$ is $\varepsilon$-approximately envy-free up to one good $(\varepsilon\text{-}\EFone{})$ for the associated fair division instance $\langle [n], [m], \V \rangle$.
\end{lemma}
\begin{proof}
Since $\x$ is $\varepsilon$-$\pEFone{}$ with respect to the price vector $\p$, for any pair of buyers $i,k \in [n]$, there exists a good $j \in \x_k$ such that $(1 + \varepsilon) \p(\x_i) \geq \p\left( \x_k \setminus \{j\} \right)$. Multiplying both sides by the maximum bang per buck ratio $\alpha_i$ of agent $i$, we get
\begin{alignat*}{4}
	&& \alpha_i \cdot (1 + \varepsilon) \p(\x_i) && \; \geq \; & \alpha_i \cdot \p\left( \x_k \setminus \{j\} \right)&&\\
	\implies \; && (1 + \varepsilon) v_i(\x_i) && \; \geq \; & \alpha_i  \cdot  \p\left( \x_k \setminus \{j\} \right)&& \quad \text{(since $\x_i \subseteq \MBB_i$)}\\
	\implies \; && (1 + \varepsilon) v_i(\x_i) && \; \geq \; & v_i \left( \x_k \setminus \{j\} \right),&&
\end{alignat*}
which is the $\varepsilon$-\EFone{} guarantee for the allocation $\x$.
\end{proof}

In order to construct the desired Fisher market equilibrium, our algorithm starts with a welfare-maximizing allocation $\x$ and a price vector $\p$ such that $\x$ is \fPO{} and each agent gets a subset of its $\MBB$ goods (this is Phase 1 of \Alg{}). If the allocation $\x$ is $3\varepsilon$-$\pEFone{}$ with respect to $\p$, then the algorithm terminates with the output $(\x,\p)$. Otherwise, the algorithm proceeds to the next phase.

In Phase 2, the algorithm works with the hierarchy of the least spending agent, and performs a series of exchanges (or swaps) of goods between the agents in the hierarchy (without changing the prices). The swaps are aimed at ensuring that at the end of Phase 2, no agent in the hierarchy is $\varepsilon$-$\pEFone{}$ envied by the least spender. Furthermore, all exchanges in Phase 2 happen only along the $\MBB{}$ edges, thus maintaining at each stage the condition that $\x$ is an equilibrium allocation, and hence, \fPO{}.

If, at the end of Phase 2, the current allocation $\x$ is still not $3\varepsilon$-$\pEFone{}$ with respect to the price vector $\p$, the algorithm moves to Phase 3. This phase consists of uniformly raising the prices of the goods owned by the members of the hierarchy. The prices are raised until either the allocation $\x$ becomes $3\varepsilon$-$\pEFone{}$ with respect to the new price vector $\p$, or a new agent gets added to the hierarchy. In the latter case, the algorithm goes back to the start of Phase 2. 

It is relevant to note that establishing the time complexity of this algorithm is an intricate task; a priori, it is not even clear whether the algorithm terminates. The stated running time bound is in fact obtained via a number of involved arguments which, in particular, rely on analyzing the spending of the agents in different phases.

\renewcommand{\floatpagefraction}{.8}
\begin{algorithm}
 \DontPrintSemicolon
 \KwIn{An instance $\I = \langle [n],[m],\V \rangle$ such that valuations are power-of--$(1+\varepsilon)$.}
 \KwOut{An integral allocation $\x$ and a price vector $\p$.}
 \BlankLine
 \tcp{---------------------------------------------------------------------------------------------Phase 1: Initialization---------------------------------------------------------------------------------------------}
 \BlankLine
 \tikzmk{A}
 $\x \leftarrow $ Welfare-maximizing allocation (allocate each good $j$ to the agent $i \in \arg\max_{k \in [n]} v_{k,j}$)\;
 $\p \leftarrow $ For each good $j \in [m]$, set $p_j = v_{i,j}$ if $j \in \x_i$.\;
 \lIf{$(\x,\p)$ is $3\varepsilon$-$\pEFone{}$\label{algline:TerminatePhase1}}{\KwRet{$(\x,\p)$}}
 \nonl \tikzmk{B}
 \boxit{mygray}
 \BlankLine
 \tcp{------------------------------------------------------------Phase 2: Removing price-envy within hierarchy------------------------------------------------------------}
 \BlankLine
 \oldnl \tikzmk{A}
 $i \leftarrow $ least spender under $(\x,\p)$\tcc*{break ties lexicographically}\label{algline:Refresh_LeastSpender}
 $\H_i \leftarrow \BuildHierarchy(i,\x,\p)$\;
 $\ell \leftarrow 1$\;
 \While{$\H_i^{\ell}$ is non-empty and $(\x,\p)$ is not $3\varepsilon$-$\pEFone{}$}{
 		\uIf{$h \in \H_i^{\ell}$ is an $\varepsilon$-path-violator  along the alternating path $P = \{i,j_1,h_1,\dots,j_{\ell - 1},h_{\ell - 1},j,h\}$}{
			$\x_h \leftarrow \x_h \setminus \{j\}$ and $\x_{h_{\ell - 1}} \leftarrow \x_{h_{\ell - 1}} \cup \{j\}$\tcc*{Swap operation}
 			Repeat Phase 2 starting from Line~\ref{algline:Refresh_LeastSpender}
 		}\Else{$\ell \leftarrow \ell + 1$}
 }
 \uIf{$(\x,\p)$ is $3\varepsilon$-$\pEFone{}$}{\KwRet{$(\x,\p)$}}\Else{Move to Phase 3 starting from Line~\ref{algline:Start_Of_Phase_3}}
 \BlankLine
 \BlankLine
 \nonl \tikzmk{B}
 \boxit{mygray}
 \BlankLine
 \tcp{-------------------------------------------------------------------------------------------------Phase 3: Price-rise-------------------------------------------------------------------------------------------------}
 \oldnl \tikzmk{A} 
 $\alpha_1 \leftarrow \min\limits_{h \in \H_i, \, j \in [m] \setminus \x_{\H_i}} \frac{\alpha_h}{v_{h,j}/p_j}$, where $\alpha_h$ is the maximum bang per buck ratio for agent $h$, and $\x_{\H_i}$ is the set of goods currently owned by members of the hierarchy $\H_i$\label{algline:Start_Of_Phase_3}\;
 \BlankLine
 \tcc*{$\alpha_1$ corresponds to raising prices until a new agent gets added to the hierarchy}
 \BlankLine
 $\alpha_2 \leftarrow \frac{1}{\p(\x_i)} \max\limits_{k \in [n] \setminus \H_i} \min\limits_{j \in \x_k} \p(\x_k \setminus \{j\})$\;
  \BlankLine
  \tcc*{$\alpha_2$ corresponds to raising prices until the $\pEFone$ condition is satisfied}
 \BlankLine
 $\alpha_3 \leftarrow (1+\varepsilon)^s$, where $s$ is the smallest integral power of $(1+\varepsilon)$ such that $(1+\varepsilon)^s > \frac{ \p(\x_h)}{\p(\x_i)}$; here $i$ is the least spender and $h \in \arg\min_{k \in [n] \setminus \H_i} \p(\x_k)$.\label{algline:alpha_3}\;
  \BlankLine
  \tcc*{$\alpha_3$ corresponds to raising prices in multiples of $(1+\varepsilon)$ until the identity of the least spender changes}
 \BlankLine
 $\alpha \leftarrow \min(\alpha_1,\alpha_2,\alpha_3)$\label{algline:alpha_defn}\;
 \ForEach{good $j \in \x_{\H_i}$}{
 	$p_j \leftarrow \alpha \cdot p_j$
 }
 \uIf{$\alpha = \alpha_2$\label{algline:alpha_2}}{\KwRet{$(\x,\p)$}}\Else{Repeat Phase 2 starting from Line~\ref{algline:Refresh_LeastSpender}}
 \BlankLine
 \BlankLine
 \nonl \tikzmk{B}
 \boxit{mygray}
 \caption{\Alg}
 \label{alg:Main}
\end{algorithm}

\section{Proof of Theorem~\ref{THM:EF1+PO_BOUNDEDVALUATIONS_PSEUDOPOLYTIME}}
\label{sec:MainResult_EFonePO_Pseudopoly}

This section presents the analysis of our algorithm and a proof of \Cref{THM:EF1+PO_BOUNDEDVALUATIONS_PSEUDOPOLYTIME}. \Cref{subsec:ALG_analysis_powersofr} presents the analysis of our algorithm for valuations that satisfy the power-of-$(1+\varepsilon)$ property. \Cref{subsec:ALG_analysis_generalcase} extends this analysis to general valuations, culminating in the proof of \Cref{THM:EF1+PO_BOUNDEDVALUATIONS_PSEUDOPOLYTIME}.

\subsection[Analysis of \Alg{} when the Valuations are power-of-r]{Analysis of \Alg{} when the Valuations are power-of-$(1+\varepsilon)$}
\label{subsec:ALG_analysis_powersofr}

In this section, we will analyze \Alg{} under the assumption that 
 all valuations are \emph{power-of-$(1+\varepsilon)$}, i.e., there exists $\varepsilon > 0$ such that for each agent $i \in [n]$ and each good $j \in [m]$, we have $v_{i,j} \in \{0,(1+\varepsilon)^a\}$ for some natural number $a$ (possibly depending on $i$ and $j$). We will start by defining the notion of a \emph{time step} that will be useful in the subsequent analysis. 

\paragraph{Time steps and events}\sloppypar
The execution of \Alg{} can be described in terms of the following four \emph{events}: (1) \emph{Swap} operation in Phase 2, (2) \emph{Change} in the identity of least spender in Phase 2, (3) \emph{Price-rise} by a factor of $\alpha$ in Phase 3, 
 and (4) \emph{Termination step}. We use the term \emph{time step} (or simply a \emph{step}) to denote the indexing of any execution of \Alg{}, e.g., \Alg{} might perform a swap operation on the first and second time steps, followed by a price-rise in the third time step, and so on. We will use the phrase ``at time step $t$'' to denote the state of the algorithm \emph{before} the event at time step $t$ takes place. Notice that each event stated above runs in polynomial time, and therefore it suffices to analyze the running time of \Alg{} in terms of the total number of events (or time steps).
 
 We will now proceed to analyzing the correctness (\Cref{lem:ALG_outputs_eps-EF1_and_fPO}) and the running time (\Cref{lem:ALG_RunningTime_PowersOfr}) of $\Alg{}$ for power-of-$(1+\varepsilon)$ valuations.

\begin{restatable}[\textbf{Correctness of \Alg{} for power-of--$(1+\varepsilon)$ instance}]{lemma}{ALGCorrectness}
 \label{lem:ALG_outputs_eps-EF1_and_fPO}
 Given any power-of-$(1+\varepsilon)$ instance as input, the allocation returned by \Alg{} is $3\varepsilon$-approximately envy-free up to one good ($3\varepsilon$-$\EFone$) and fractionally Pareto efficient ($\fPO$).
\end{restatable}
\begin{proof}
Let the output of \Alg{} be $(\x,\p)$. The fact that $\x$ is \fPO{} follows from the observation that at each step of the algorithm, the allocation of any agent is a subset of its $\MBB{}$ goods, i.e., at each time step, we have $\x_i \subseteq \MBB_i$ for each agent $i \in [n]$. This is certainly true at the end of Phase 1 by way of setting the prices. In Phase 2, each swap operation only happens along an alternating $\MBB{}$-allocation edge, which maintains the $\MBB{}$ condition. Phase 3 involves raising the prices of the goods owned by the members of the hierarchy $\H_i$ without changing the allocation. We will argue that for each agent $k \in [n]$, if $\x_k \subseteq \MBB_k$ before the price-rise, then the same continues to hold after the price-rise. Indeed, for any agent $k \notin \H_i$, we have $\x_k \cap \x_{\H_i} = \emptyset$. As a result, raising the prices of the goods in $\x_{\H_i}$ does not affect the bang per buck ratio of agent $k$ for the goods in $\x_k$ (and can only reduce its bang per buck ratio for the goods in $\x_{\H_i}$), thus maintaining the above condition. For any agent $k \in \H_i$, we have $\MBB_k \subseteq \x_{\H_i}$ by construction of the hierarchy. Raising the prices of the goods in $\x_{\H_i}$ therefore corresponds to lowering the $\MBB$ ratios for the agents in $\H_i$. By choice of $\alpha_1$, the price-rise stops as soon as a new $\MBB{}$-edge appears between an agent $k \in \H_i$ and a good $j \notin \x_{\H_i}$. This ensures that the \emph{new} maximum bang per buck ratio for any agent $k \in \H_i$ does not fall below its second highest bang per buck ratio prior to the price-rise, thus guaranteeing $\x_k \subseteq \MBB_k$. 

We can now define a Fisher market where each agent's endowment equals its spending under $\x$. Since $(\x,\p)$ is an equilibrium for this market, we have that $\x$ is \fPO{} (\Cref{prop:FirstWelfareTheorem}).

Next, we will argue that $\x$ is $3\varepsilon\text{-}\EFone{}$. Notice that \Alg{} terminates only if either the current outcome $(\x,\p)$ is $3\varepsilon\text{-}\pEFone{}$, or when $\alpha = \alpha_2$ (Line~\ref{algline:alpha_2}). In the first case, we get that $\x$ is $3\varepsilon\text{-}\EFone{}$ for the underlying fair division instance (\Cref{lem:eps-pEF1_implies_eps-EF1}).  Therefore, we only need to analyze the second case.

 Let us suppose that the termination happens at time step $t$, and let $\q$ be the price vector maintained by $\Alg$ just \emph{before} the price-rise step that lead to termination. After the time step $t$, $\Alg$ terminates with the allocation $\x$ and price vector $\p$. Since Phase 3 does not change the ownership of the goods, the allocation maintained by $\Alg$ just before termination is also $\x$. 
 
 Let $i$ be the least spender at time step $t$, and let $\H_i$ be the hierarchy of agent $i$. Since Phase 3 only affects the prices of the goods in $\x_{\H_i}$, we have that $\p{(\x_k)} = \q{(\x_k)}$ for all $k \in [n] \setminus \H_i$, and $\p{(\x_k)} = \alpha_2 \q{(\x_k)}$ for all $k \in \H_i$. Additionally, at the end of (any execution of) Phase 2, no agent in the least spender's hierarchy is an $\varepsilon$-path-violator (and hence is also not an $\varepsilon$-violator). Thus,
\begin{align}
\label{eqn:No_Eps_Price_Envy_Within_Hierarchy}
	& (1+\varepsilon) \q(\x_i) \geq \max_{k \in \H_i} \min_{j \in \x_k} \q(\x_k \setminus \{j\}) \nonumber \\
	\implies & (1+\varepsilon) \p(\x_i) \geq \max_{k \in \H_i} \min_{j \in \x_k} \p(\x_k \setminus \{j\}).
\end{align}

By definition of $\alpha_2$, we have the following condition for the agents outside the hierarchy:
\begin{align}
\label{eqn:No_Price_Envy_Outside_Hierarchy}
	\p(\x_i) = \alpha_2 \q(\x_i) \geq \max_{k \in [n] \setminus \H_i} \min_{j \in \x_k} \q(\x_k \setminus \{j\}) = \max_{k \in [n] \setminus \H_i} \min_{j \in \x_k} \p(\x_k \setminus \{j\}).
\end{align}
 
\Cref{eqn:No_Eps_Price_Envy_Within_Hierarchy,eqn:No_Price_Envy_Outside_Hierarchy} together imply that
\begin{align}
\label{eqn:No_Price_Envy_Overall}
	(1+\varepsilon) \p(\x_i) \geq \max_{k \in [n]} \min_{j \in \x_k} \p(\x_k \setminus \{j\}),
\end{align}
 which means that the outcome $(\x,\p)$ is $\varepsilon$-$\pEFone$ for agent $i$. If agent $i$ is a least spender under $(\x,\p)$ (i.e., $i$ continues to a least spender after the price rise), then $\x$ is $\varepsilon\text{-}\pEFone$ with respect to $\p$, and the lemma follows. Otherwise, an agent $h \in \arg\min_{k \in [n] \setminus \H_i} \q(\x_k)$ must become the least spender after the final price-rise step. In this case, we have that 
 \begin{align*}
 	 (1+\varepsilon) \q(\x_h) & \geq \alpha_3 \q(\x_i) \quad \text{(by definition of }\alpha_3) \\
 	\implies  (1+\varepsilon) \q(\x_h) & \geq \alpha_2 \q(\x_i) \quad (\alpha = \alpha_2 \implies \alpha_2 \leq \alpha_3) \\
 	\implies  (1+\varepsilon) \p(\x_h) & \geq \p(\x_i) \qquad (\text{since }\p(\x_i) = \alpha_2 \q(\x_i) \text{ and } \p(\x_h) = \q(\x_h))\\
 	\implies  (1+\varepsilon)^2 \p(\x_h) & \geq \min_{j \in \x_k} \p(\x_k \setminus \{j\}) \qquad \text{ for all } k \in [n],
 \end{align*}
 where the last inequality follows from \Cref{eqn:No_Price_Envy_Overall}. Since $0 < \varepsilon < 1$, we have that $(1+\varepsilon)^2 < 1+3\varepsilon$. Thus, the new least spender (agent $h$) satisfies $(1+3\varepsilon) \p(\x_h) \geq \min_{j\in\x_k}\p(\x_k\setminus\{j\})$ for all $ k \in [n]$. This implies that $(\x,\p)$ is $3\varepsilon\text{-}\pEFone$. The stated claim now follows from \Cref{lem:eps-pEF1_implies_eps-EF1}.
\end{proof}

\begin{restatable}[\textbf{Running time bound for power-of--$(1+\varepsilon)$ instance}]{lemma}{ALGRunningTime}
 \label{lem:ALG_RunningTime_PowersOfr}
 Given any power-of-$(1+\varepsilon)$ instance as input, \Alg{} terminates in time $\O \left( \poly(m,n,\frac{1}{\varepsilon},\ln v_{\max}) \right)$, where $v_{\max} = \max_{i,j} v_{i,j}$.
\end{restatable}

The proof of \Cref{lem:ALG_RunningTime_PowersOfr} appears in \Cref{subsec:Proof_Of_ALG_RunningTime_PowersOfr}.

\subsection{Analysis of \Alg{} for General Valuations: Proof of Theorem~\ref{THM:EF1+PO_BOUNDEDVALUATIONS_PSEUDOPOLYTIME}}
\label{subsec:ALG_analysis_generalcase}

In this section, we will show that for any given fair division instance $\I = \langle [n], [m], \V \rangle$ with integral valuations, an allocation that is envy-free up to one good $(\EFone{})$ and Pareto efficient $(\PO{})$ can be found in pseudopolynomial time (\Cref{THM:EF1+PO_BOUNDEDVALUATIONS_PSEUDOPOLYTIME}). In particular, for bounded valuations, this result provides a polynomial time algorithm for computing an $\EFone{}$ and $\PO$ allocation (\Cref{rem:EF1+PO_BoundedValuations_PolyTime}).

We will prove \Cref{THM:EF1+PO_BOUNDEDVALUATIONS_PSEUDOPOLYTIME} by running \Alg{} on a \emph{$\varepsilon$-rounded version} $\I' = \langle [n], [m], \V' \rangle$ of the given instance $\I$ for some parameter $\varepsilon > 0$. The instance $\I'$ is a power-of-$(1+\varepsilon)$ instance\footnote{Recall that in a \emph{power-of-$(1+\varepsilon)$} instance, we have $v_{i,j} \in \{0,(1+\varepsilon)^a\}$ for some $a \in \N$ (possibly depending on $i$ and $j$).} constructed by \emph{rounding up} the valuations in $\I$ to the nearest integer power of $(1+\varepsilon)$. From \Cref{lem:ALG_outputs_eps-EF1_and_fPO}, we know that the allocation returned by $\Alg{}$ is $3\varepsilon$-$\EFone{}$ and $\fPO$ with respect to $\I'$ (for any given $\varepsilon > 0$). We will show that for an appropriate choice of $\varepsilon$, the same allocation turns out to be $\EFone{}$ and $\PO{}$ with respect to the original instance $\I$. In addition, the running time bound in \Cref{lem:ALG_RunningTime_PowersOfr} instantiated for this choice of $\varepsilon$ will show that $\Alg$ runs in pseudopolynomial time.

More formally, the \emph{$\varepsilon$-rounded version} $\I' = \langle [n], [m], \V' \rangle$ of the given instance $\I$ is constructed as follows: For each agent $i \in [n]$ and each good $j \in [m]$, the valuation $v'_{i,j}$ is given by
$$ v'_{i,j}:=
\begin{cases} 
(1+\varepsilon)^{\lceil \log_{1+\varepsilon}v_{i,j} \rceil} & \text{ if } v_{i,j}> 0, \\
0 & \text{ if } v_{i,j}= 0. \\
\end{cases}$$

Notice that $v_{i,j} \leq v'_{i,j} \leq (1+\varepsilon)v_{i,j}$ for each agent $i$ and each good $j$.

\Cref{lem:Small_delta_PO} establishes that for an appropriate choice of $\varepsilon$, an allocation that is $\fPO$ for the $\varepsilon$-rounded instance $\I'$ is $\PO$ with respect to the original instance $\I$. The proof of \Cref{lem:Small_delta_PO} appears in \Cref{subsec:OmittedProofs_ALG_analysis_generalcase}.

\begin{restatable}{lemma}{SmallDeltaPO}
\label{lem:Small_delta_PO}
Let $\I = \langle [n], [m], \V \rangle$ be a fair division instance, and let $\varepsilon \leq \frac{1}{6 m^3 v_{\max}^4}$. Then, an allocation $\x$ that is $\fPO$ for $\I'$ (the $\varepsilon$-rounded version of $\I$) is $\PO$ for the original instance $\I$.
\end{restatable}

\Cref{lem:Small_epsilon_int} establishes that for a small enough $\delta$, a $\delta$-$\EFone$ allocation is in fact $\EFone$.

\begin{restatable}{lemma}{ExactEFoneSmallEpsilon} 
\label{lem:Small_epsilon_int}
Let $\I = \langle [n], [m], \V \rangle$ be a fair division instance, and let $0 < \delta \leq \frac{1}{2 m v_{\max}}$. Then, an allocation $\x$ is $\delta\text{-}\EFone{}$ for $\I$ if and only if it is $\EFone{}$ for $\I$. 
\end{restatable}
\begin{proof}
If $\x$ is $\delta\text{-}\EFone{}$, we have that for every pair of agents $i,k\in [n]$, there exists a good $j \in \x_k$ such that $(1+\delta) v_i(\x_i) \geq v_i(\x_k\setminus\{j\})$. The bound on $\delta$ implies that $v_i(\x_k\setminus\{j\}) - v_i(\x_i) \leq \frac{1}{2}$. Integrality of valuations gives $v_i(\x_k\setminus\{j\}) - v_i(\x_i) \leq 0$, as desired.  
\end{proof}

\EFonePOBoundedValuationsPseudoPolyTime*
\begin{proof}
Let $\I' = \langle [n], [m], \V' \rangle$ be the $\varepsilon$-rounded version of $\I$ with $\varepsilon = \frac{1}{14 m^3 v_{\max}^4}$. From \Cref{lem:ALG_outputs_eps-EF1_and_fPO,lem:ALG_RunningTime_PowersOfr}, we know that an allocation $\x$ that is $3\varepsilon\text{-}\EFone{}$ and $\fPO{}$ for $\I'$ can be found in $\O\left( \poly( m,n,\frac{1}{\varepsilon}, \ln v_{\max} ) \right)$ time. Under the stated choice of $\varepsilon$, \Cref{lem:Small_delta_PO} implies that $\x$ must be Pareto efficient ($\PO$) for $\I$. Therefore, we only need to show that $\x$ is $\EFone{}$ for the instance $\I$.

Since $\x$ is $3\varepsilon\text{-}\EFone{}$ for $\I'$, we have that for every pair of agents $i,k \in [n]$, there exists a good $j \in \x_k$ such that
$(1+3\varepsilon) v'_i(\x_i) \geq v'_i(\x_k\setminus\{j\})$. Furthermore, since $\I'$ is a $\varepsilon$-rounded version of $\I$, we have that $v'_{i,j}\leq (1+\varepsilon) v_{i,j}$ for each good $j\in[m]$. Hence, $(1+\varepsilon)(1+3\varepsilon) v_i(\x_i) \geq v'_i(\x_k\setminus\{j\})$. Finally, since the valuations in $\I'$ are a rounded-up version of those in $\I$, we have that $v_{i,j}\leq v'_{i,j}$ for each good $j\in[m]$, and thus $(1+\varepsilon)(1+3\varepsilon) v_i(\x_i) \geq v_i(\x_k\setminus\{j\})$. For $\varepsilon \leq 1$, this expression simplifies to $(1+7\varepsilon) v_i(\x_i) \geq v_i(\x_k\setminus\{j\})$, which means that $\x$ is $7\varepsilon$-$\EFone{}$ for the instance $\I$. Instantiating \Cref{lem:Small_epsilon_int} for $\delta = 7\varepsilon$ gives that $\x$ is $\EFone{}$ for $\I$. 
\end{proof}

\begin{restatable}{lemma}{ApproxEFApproxPOPolytime}
\label{lem:Approx_EF_Approx_PO_Polytime}
Given the $\varepsilon$-rounded version $\I'$ (of the instance $\I$) as input, \Alg{} finds a $7\varepsilon$-$\EFone{}$ and $\varepsilon$-$\PO$ allocation for $\I$ in $\O(\poly(m,n,\frac{1}{\varepsilon}, \ln v_{\max}))$ time.
\end{restatable}
\begin{proof}
Let $\x$ be the allocation returned by \Alg{}. From \Cref{lem:ALG_outputs_eps-EF1_and_fPO}, we know that $\x$ is $3\varepsilon$-$\EFone$ and $\fPO$ for the $\varepsilon$-rounded instance $\I'$. By an argument similar to the one in the proof of \Cref{THM:EF1+PO_BOUNDEDVALUATIONS_PSEUDOPOLYTIME}, this implies that $\x$ is $7\varepsilon$-$\EFone$ for the original instance $\I$. The running time guarantee follows from \Cref{lem:ALG_RunningTime_PowersOfr}. Hence, we only need to show that $\x$ is $\varepsilon$-$\PO$.

Suppose, for contradiction, that $\x$ is $\varepsilon$-Pareto dominated by an allocation $\y$. Thus, $v_k(\y_k) \geq (1+\varepsilon) v_k(\x_k)$ for every agent $k \in [n]$ and $v_i(\y_i) > (1+\varepsilon) v_i(\x_i)$ for some agent $i \in [n]$. By construction of the $\varepsilon$-rounded instance $\I'$, we know that $v_{k,j} \leq v'_{k,j} \leq (1+\varepsilon)v_{k,j}$ for each agent $k$ and each good $j$. Using the inequality $v'_{k,j} \leq (1+\varepsilon)v_{k,j}$ in a good-by-good manner for the bundle $\x_k$, along with the additivity assumption of valuations in the instance $\I$, we get that $(1+\varepsilon) v_k(\x_k) \geq v'_k(\x_k)$. By a similar application of the inequality $v_{k,j} \leq v'_{k,j}$ for the bundle $\y_k$, we get $v'_k(\y_k) \geq v_k(\y_k)$. Combining these relations gives $v'_k(\y_k) \geq v'_k(\x_k)$ for every agent $k \in [n]$ and $v'_i(\y_i) > v'_i(\x_i)$ for some agent $i \in [n]$. However, this means that the allocation $\y$ Pareto dominates the allocation $\x$ in the instance $\I'$, which is a contradiction since $\x$ is $\fPO$ for $\I'$.
\end{proof}

\section{Existence Result: Proof of Theorem~\ref{thm:EF1+fPO_Existence}}
\label{sec:EF1+fPO_Existence}

\EFonefPOExistence*
\begin{proof}

Given any fair division instance $\I = \langle [n], [m], \V \rangle$, define $\varepsilon_z \coloneqq \frac{1}{14 z m^3 v_{\max}^4}$ for any natural number $ z \in \N$. Write $\I^{z} = \langle [n], [m], \V^{z} \rangle$ to denote the $\varepsilon_z$-rounded version of $\I$, with $\V^{z}=\{v_{1}^z, v_{2}^z, \ldots,v_n^z \}$ being the set of rounded valuations. By instantiating \Cref{lem:ALG_outputs_eps-EF1_and_fPO} with $\varepsilon=\varepsilon_z$, we can ensure that each rounded version $\I^z$ admits an allocation $\x^z$ which is $\varepsilon_z\text{-}\EFone{}$ and $\fPO$.

Note that for all $z \geq 1$, allocation $\x^z$ is an $\EFone{}$ allocation for the instance $\I$. This follows from the analysis of Theorem~\ref{THM:EF1+PO_BOUNDEDVALUATIONS_PSEUDOPOLYTIME}, wherein we showed that if an allocation $\x$ is $\varepsilon_z\text{-}\EFone{}$ with respect to $\I^z$ (the $\varepsilon_z$-rounded version of $\I$), then $\x$ is $\EFone{}$ for $\I$, as long as $\varepsilon_z\leq \frac{1}{14 m^3 v_{\max}^4}$. 

To complete the proof, we will show that in the sequence of $\EFone{}$ allocations $(\x^z)_{z \in \N}$, there exists one which is also $\fPO{}$ for $\I$. In particular, by applying the second welfare theorem of Fisher markets (\Cref{thm:SecondWelfareTheoremFisherMarkets} in \Cref{subsec:Second_Welfare_Theorem_for_Fisher_Markets}), we get that for every $z$, there exists a price vector $\p^z$ such that $\x^z$ satisfies the $\MBB$ condition with respect to $\p^z$ and the valuations in $\I^z$. That is, $ \x^z_i \subseteq \MBB^{z}_i \coloneqq \left\{j \in [m] \, : \, \frac{v_i^z (j)}{\p^z_j} = \alpha_i^z \right\}$ for each agent $i \in [n]$; here, $\alpha_{i}^z$ is the maximum bang per buck ratio of agent $i$ with respect to the price vector $\p^z$, i.e., $\alpha_{i}^z \coloneqq \max_j \frac{v_{i}^z (j)}{p_{j}}$.

Observe that if the outcome $(\x^z,\I^z,\p^z)$ satisfies the $\MBB$ condition, then so does $(\x^z,\I^z,c\cdot\p^z)$, for any positive $c$. In particular, by scaling $\p^z$ by a factor of $c = \frac{1}{\max_{j\in[m]}\p^z_j}$, we ensure that the price vector $ c \cdot \p^z $ lies in $[0,1]^m$  for all $z$. Therefore, without loss of generality, we will assume that $\p^z$ is bounded for all $z$. In addition, since there are only finitely many integral allocations of $[m]$, there exists an infinite sequence $ z_1, z_2,\dots$ such that $\x^{z_a}=\x^{z_b}$ for all $a,b \in \N$. Let $\x$ denote the allocation $\x^{z_a}$, i.e., $\x:=\x^{z_a}$  for all $a \in \N$. We know from the Bolzano-Weierstrass theorem that there exists a subsequence of $(\p^{z_1},\p^{z_2}, \ldots)$ and a price vector $\p^*$ such that the subsequence converges to $\p^*$, i.e., $\lim_{a \rightarrow \infty} \p^{z_{q_a}}=\p^*$. Let this subsequence be $(\p^{z_{q_1}}, \p^{z_{q_2}}, \dots)$.

Overall, by a relabeling of the indices, we can guarantee the existence of a sequence of natural numbers $(z_1,z_2, \ldots)$ such that $\x^{z_a}=\x$ for all $a \in \N$ and $\lim_{a\rightarrow\infty} \p^{z_a}=\p^*$. In particular, we have that $\x$ is an $\fPO{}$ allocation for each instance in the sequence $\I^{z_1}, \I^{z_2}, \ldots$, and hence it satisfies the $\MBB$ condition in these instances as well. 

Since $\I^{z_a}$ is an $\varepsilon_{z_a}$-rounded version of $\I$, for any agent $i \in [n]$ and good $j \in [m]$, we have that $v_{i,j} \leq v_{i,j}^{z_a} \leq (1+\varepsilon_{z_a}) v_{i,j}$. Hence, from the sandwich theorem, we get that $v_{i,j}^{z_a}$ converges to $v_{i,j}$ as $a$ tends to infinity, i.e., $\lim_{a\rightarrow\infty}v_{i,j}^{z_a} = v_{i,j}$.

Using the fact that $\x$ satisfies the $\MBB$ condition with respect to $(\p^{z_a},\I^{z_a})$, we will show that it satisfies the $\MBB$ condition for $(\p^*, \I)$ as well. Specifically, consider any agent $i \in [n]$ along with any pair of goods $j \in [m]$ and $j'\in \x_i$ such that $v_{i,j}^{z_a}$ is nonzero, and hence $v_{i,j}$ is nonzero.\footnote{Since $\I^{z_a}$ is a $\varepsilon_{z_a}$-rounded version of $\I$, $v_{i,j}{z_a}$ is zero if and only if $v_{i,j}$ is zero.} Now, the $\MBB$ condition in $\I^{z_a}$ implies that
 \begin{alignat*}{4}
	  && \frac{\p^{z_a}_{j'}}{v_{i,j'}^{z_a}} && \; \leq \; & \frac{\p^{z_a}_{j}}{v_{i,j}^{z_a}}.&&
	 \end{alignat*}
	 Under the limit of $a$ approaching infinity, we get
	 \begin{alignat*}{4}
	 && \lim_{a\rightarrow\infty}\frac{\p^{z_a}_{j'}}{v_{i,j'}^{z_a}} && \; \leq \; & \lim_{a\rightarrow\infty}\frac{\p^{z_a}_{j}}{v_{i,j}^{z_a}} &&\\
	 \implies \; && \frac{\p^{*}_{j'}}{v_{i,j'}} && \; \leq \; & \frac{\p^{*}_{j}}{v_{i,j}}.&&
	 \end{alignat*}
	 
Therefore, the allocation $\x$ satisfies the $\MBB$ condition with respect to the price vector $\p^*$ and the fair division instance $\I$. Hence, using the first welfare theorem (\Cref{prop:FirstWelfareTheorem}), we get that $\x$ is $\fPO$ for $\I$. This completes the proof of \Cref{thm:EF1+fPO_Existence}.
\end{proof}

\section{Nash Social Welfare Approximation: Proof of Theorem~\ref{THM:APPROXNASH}}
\label{sec:ApproxNash}
This section proves that \Alg{} provides a $1.45$ ($\approx e^{1/e}$)-approximation for the Nash social welfare maximization problem in polynomial time. We begin by showing (in \Cref{lemma:id-val}) that if the valuations of all the agents are identical, then any $\varepsilon$-$\EFone{}$ allocation provides a $e^{(1+\varepsilon)/e}$-approximation to Nash social welfare. We will then use this result, along with an appropriate choice of $\varepsilon$, to prove the desired approximation bound in \Cref{THM:APPROXNASH}.

\EFoneNSWIdenticalVals*

The proof of \Cref{lemma:id-val} makes use a structural result stated as \Cref{lem:opt-struct}. A relevant notion used in these results is that of \emph{partially-fractional allocations}, defined as follows: Given a subset $B \subset [m]$ of the set of goods, a partially-fractional allocation (with respect to $B$) is one where the goods in $B$ have to be integrally allocated and the remaining goods in $[m]\setminus B$ can be fractionally allocated. Formally, a partially-fractional allocation $\y \in [0,1]^{n \times m}$ is such that for each agent $i \in [n]$, we have $\y_{i,j} \in \{0, 1\}$ for each $j \in B$, and $\y_{i,j} \in [0,1]$ for each $j \in [m]\setminus B$. We write $\F_B$ to denote the set of all partially-fractional allocations with respect to $B$.


\begin{restatable}{lemma}{NashOptStruct}
\label{lem:opt-struct}
Let $\I = \langle [n], [m], \V \rangle$ be an instance with additive and identical valuation functions (denoted by $v$ for each agent) such that $m \geq n$. Let $B \subset [m]$ be a subset of goods such that $|B| < n$. Then, there is a partially-fractional allocation $\omega = (\omega_1, \ldots, \omega_n) \in \F_B$ that maximizes Nash social welfare (among allocations in $\F_B$) such that
\begin{enumerate}[(1)]
\item Each agent gets at most one good from $B$ under $\opt$. That is, for each agent $i \in [n]$, we have that $\textstyle{ | \{ j \in [m] \, : \, \opt_{i,j} > 0 \} \cap B | \leq 1 }$.
\item Any agent with strictly-better-than-the-worst allocation under $\opt$ gets exactly one integral good (and no fractional good). That is, for any agent $i \in [n]$ such that $v(\opt_i) > \min_k v(\opt_k)$, we have $\opt_{i,j} = 1 $ for some $j \in B$, and $\opt_{i,j'} = 0$ for all $j' \neq j$.
\end{enumerate}
\end{restatable}
\begin{proof}(of \Cref{lem:opt-struct})
We will consider a partially-fractional allocation $\opt= (\opt_1, \opt_2, \ldots, \opt_n ) \in \F_B$ that maximizes the Nash social welfare over $\F_B$, and show that it can be transformed (without decreasing NSW) into an allocation in $\F_B$ that satisfies the stated properties.

Observe that if an agent $i \in [n]$ receives a fractional good under $\opt$ (i.e., $\opt_{i,j'} >0$ for some $j' \notin B$), then its value $v(\opt_i)$ is equal to $\mu \coloneqq \min_k v(\opt_k)$. This is because if $v(\opt_i) > \mu$, then we can ``redistribute'' $j'$ between agent $i$ and the least valued agent $\arg\min_{k} v(\opt_k)$ to obtain another fractional allocation in $\F_B$ with Nash social welfare strictly greater than that of $\opt$. This contradicts the optimality of $\opt$.

Therefore, any agent with value strictly greater than $\mu$ can only receive integral goods. We will show that any such agent necessarily receives exactly one integral good, and hence establish property (2). Suppose, for contradiction, that some agent $i \in [n]$ receives distinct goods $a,b \in B$ (i.e., $\opt_{i,a} = \opt_{i,b} = 1$), and $v(\opt_i) > \mu$. Since $|B| < n$, there exists some agent $k \in [n]$ that does not receive any integral good from $B$. Since agent $k$ only receives fractional goods under $\opt$, it follows from the above argument that $v(\opt_k) = \mu$. We can now assign one of the integral goods (say $a$) to agent $k$, and redistribute the fractional goods in $\opt_k$ between the agents $i$ and $k$ to obtain a new partially-fractional allocation with strictly greater Nash social welfare than $\opt$. This contradicts the optimality of $\opt$. Therefore, any agent with value strictly greater than $\mu$ must receive exactly one integral good under $\opt$.

Next we address property (1). Note that this property already holds for agents with value strictly greater than $\mu$. Hence, we only need to consider agents whose value is exactly equal to $\mu$. Let $i \in [n]$ be an agent that receives two distinct goods $a,b \in B$ such that $v(\opt_i) = \mu$. Since $|B| < n$, there must exist an agent $k \in [n]$ that only receives fractional goods under $\opt$. We can now perform a ``swap'' by assigning one of the integral goods (say $a$) to agent $k$, and fractional goods of total value $v(a)$ in $\opt_k$ to agent $i$ to obtain a new partially-fractional allocation in $\F_B$ with the same Nash social welfare as $\opt$. Such a swap is always possible, since $v(a) \leq v(\opt_i) = \mu = v(\opt_k)$. By repeating this process at most $n$ times, we can obtain an NSW maximizer satisfying property (1).
\end{proof}

We will now prove \Cref{lemma:id-val}.

\begin{proof}(of \Cref{lemma:id-val})
Let $\I = \langle [n], [m], \V \rangle$ denote the given instance with identical and additive valuations, and let $v$ denote the valuation function of all the agents. Write $\x = (\x_1, \x_2, \ldots, \x_n)$ to denote an $\varepsilon\text{-}\EFone{}$ allocation of $\I$. Let $\ell$ denote the value of the least valued bundle in $\x$, i.e., $\ell \coloneqq \min_{i \in [n]} v(\x_i)$. By reindexing, we have that $v(\x_1) \geq v(\x_2) \geq \ldots \geq  v(\x_{n})= \ell$.

We will use $g_i$ to denote a largest valued good in the bundle $\x_i$ of agent $i$, i.e., $g_i \in \arg\max_{g \in \x_i} v(g)$. The fact that $\x$ is an $\varepsilon\text{-}\EFone{}$ allocation implies that 
	\begin{align}
		\label{eq:EF1}
		v(\x_i \setminus \{g_i \}) \leq (1+\varepsilon)\ell \qquad \text{for all $i\in[n]$.}
	\end{align}
	
Define $B \coloneqq \{g_1,g_2,\ldots,g_{n-1}\}$. 
Let $\omega= (\omega_1, \ldots, \omega_n) \in \F_B$ be a partially-fractional allocation (with respect to $B$) that maximizes Nash social welfare among all allocations in $\F_B$. Since $\F_B$ contains all the integral allocations, we have $\NW(\omega) \geq \NW(\x^*)$, where $\x^*$ is a Nash optimal (integral) allocation. Hence, to prove the lemma, it suffices to show that $\NW(\x) \geq \frac{1}{e^{{(1+\varepsilon)}/{e}}} \NW(\omega)$.

Define $\alpha \coloneqq \min_{k\in[n]}v(\omega_k)/\ell$, and let $H \coloneqq \left\{ k \in[n] \, : \, v(\x_k)>\alpha\ell \right\}$. We will now consider partially-fractional allocations wherein only (and all) the goods in $\x_H$ have to be allocated integrally, and the remaining goods can be fractionally allocated. Write $\F_{\x_H}$ to denote the set of all such partially-fractional allocations.

The rest of the proof consists of four parts: First, we will construct an allocation $\x' \in \F_{\x_H}$ such that $\NSW(\x') \leq \NSW(\x)$. (Doing this will allow us to work with the ratio $\frac{\NSW(\x')}{\NSW(\omega)}$, which is convenient to bound from below.) Second, we will derive a lower bound on $\NSW(\x')$ in terms of the relevant parameters $\alpha$, $\ell$, and $n$ (and two other parameters $h$ and $t$ that we will define shortly). Third, we will derive an upper bound on $\NSW(\omega)$ in terms of the same parameters. Finally, we will derive relationships between these parameters in order to achieve the stated approximation ratio.

\begin{itemize}

	\item \emph{Constructing the allocation $\x'$}: We start by initializing $\x'\leftarrow\x$. While there exist two agents $i,k \in [n]$ such that $\ell < v(\x'_i) < v(\x'_k) < \alpha\ell$, we transfer goods of value $\Delta \coloneqq \min\{ v(\x'_i)-\ell,\alpha\ell - v(\x'_k) \}$ from $\x'_i$ (the lesser valued bundle) to $\x'_k$ (the larger valued bundle). In particular, this transfer of goods ensures that the Nash social welfare does not increase. Also, this process must terminate because after every iteration of the while-loop, either $v(\x'_i) = \ell$ or $v(\x'_k) = \alpha\ell$, and therefore at least one of these agents does not participate in future iterations of the while-loop. Moreover, we have $\x'_k = \x_k$ for all $k \in H$, as the agents in $H$ do not participate in the transfer. This proves that $\x' \in \F_{\x_H}$.
	
	\item \emph{Lower bound for $\NSW(\x')$}: Notice that there can be at most one agent $s$ in the allocation $\x'$ such that $v(\x'_s)\in (\ell,\alpha\ell)$. This is because the while-loop continues to execute if there are two or more such agents. For every other agent $k\in [n]\setminus H$, $v(\x'_k)$ is either $\ell$ or $\alpha\ell$. 
	
Let $h \coloneqq |\{ k \in [n] \, : \, v(\x_k) > \alpha \ell \}|$ denote the cardinality of the set $H$ (i.e., $h=|H|$). Let $t \coloneqq |\{ k \in [n] \, : \, v(\x'_k) \geq \alpha \ell \}|$ denote the number of agents with a valuation at least $\alpha\ell$ in the allocation $\x'$. Thus, there are $(n-t)$ agents with valuation strictly below $\alpha \ell$ in $\x'$. We lower bound the valuations of these agents by $\ell$ in order to obtain the following relation:
	\begin{align}
	\label{eqn:NSW_LowerBound}
	\NW(\x') & \geq \left( \prod_{i = 1}^{h} v(\x_i) \times \left( \alpha \ell \right)^{(t-h)} \times \ell^{(n-t)} \right)^{1/n}. 
	\end{align}

	\item \emph{Upper bound for $\NSW(\omega)$}: Recall that $\omega= (\omega_1, \ldots, \omega_n) \in \F_B$ is a partially-fractional allocation (with respect to $B$) that maximizes Nash social welfare, and $|B| < n$. Using \Cref{lem:opt-struct}, we can assume, without loss of generality, that $\omega$ has the following two properties: (1) each agent $i \in [n-1]$ gets the good $g_i$ under $\omega$ (this can be ensured via reindexing since the valuations are identical), and (2) if $v(\omega_i) > \min_{k\in[n]} v(\omega_k)$ for any $i \in [n]$, then agent $i$ gets exactly one integral good under $\omega$ (and no fractional good).
	
	We will now argue that $v(\omega_k)\leq v(\x_k)$ for all $k \in H$. Suppose, for contradiction, that there exists an agent $k \in H$ such that $v(\omega_k) > v(\x_k)$. By definition of $H$, $v(\x_k) > \alpha \ell$ for all $k \in H$, and therefore $v(\omega_k) > \alpha \ell$. We also know that $\min_{k\in[n]} v(\omega_k) = \alpha \ell$, and therefore, by property (2), agent $k$ must get exactly one integral good $g_k$ under $\omega$ (and no fractional good). However, since $g_k \subseteq \x_k$, this contradicts the condition $v(\omega_k) > v(\x_k)$.
	
	By a similar reasoning, we can show that $v(\omega_k) = \alpha\ell$ for all $k \in [n] \setminus H$. Indeed, if $v(\omega_k) > \alpha\ell = \min_a v(\omega_a)$ for some $k \in [n] \setminus H$, then by property (2), agent $k$ must get exactly one integral good $g_k$ under $\omega$ (and no fractional good). This would imply that $v(\x_k) \geq v(g_k) = v(\omega_k) > \alpha \ell$, which contradicts the fact that $k \in [n] \setminus H$.

These observations imply the following upper bound on the Nash social welfare of $\omega$:
	\begin{align}
	\label{eqn:NSW_UpperBound}
	\NW(\omega) & \leq \left( \prod_{i = 1}^{h} v(\x_i) \times \left( \alpha \ell \right)^{(n-h)}\right)^{1/n}. 
	\end{align}

	\item \emph{Deriving relationship between the parameters}: From \Cref{eq:EF1} and from the fact that $\x'_k=\x_k$ for all $k\in H$, we have $v(\x'_k\setminus\{g_k\})=v(\x_k\setminus\{g_k\})\leq (1+\varepsilon)\ell$ for all $k\in H$. Hence, we can upper bound the value of all goods in $[m]$ excluding the $h$ goods in the set $\bigcup_{k \in H} \{g_k\}$, as follows:
	\begin{align}
	\label{ineq:total-val}
	\sum_{i \in H} v(\x'_i \setminus \{ g_i \}) + \sum_{i \in [n]\setminus H} v( \x'_i) & \leq (1 + \varepsilon) h \ell + \alpha\ell(t-h+1) + \ell(n-t-1).
	\end{align}

	Next, we will derive a lower bound for this total value by considering $\omega$. Note that all the goods in the set $\bigcup_{k \in H} \{g_k\}$ are integrally allocated under $\omega$. Hence, at least $(n-h)$ agents do not receive any good from the set $\bigcup_{k \in H} \{g_k\}$. The cumulative value derived by these agents under $\omega$ is at least $\alpha \ell (n-h)$, since $\min_{k} v(\omega_k) = \alpha \ell$. Using \Cref{ineq:total-val} and the fact that $h\leq t$, we get 
	\begin{align*}
	\alpha \ell (n-h) & \leq (1 + \varepsilon) t \ell + \alpha\ell(t-h+1) + \ell(n-t-1).
	\end{align*}

Simplification gives that $t \geq \frac{ (\alpha -1) (n-1)} {\alpha + \varepsilon}$, which can be further simplified to obtain
	\begin{align*}
	(n-t) \leq \frac{n(1 + \varepsilon) + \alpha -1}{\alpha + \varepsilon} \leq \frac{n(1 + \varepsilon)}{\alpha} + \frac{\alpha -1}{\alpha} \leq \frac{n(1+\varepsilon)}{\alpha} +1.
	\end{align*}
\end{itemize}

Recall that $\NSW(\x) \geq \NSW(\x')$. Using the above relation with \Cref{eqn:NSW_LowerBound,eqn:NSW_UpperBound} gives
	\begin{align}
	\label{eqn:NSWApprox_LowerBound}
		\frac{\NSW(\x)}{\NSW(\omega)}\geq \frac{\NSW(\x')}{\NSW(\omega)}\geq \frac{\left( \prod_{i = 1}^{h} v(\x_i) \times \left( \alpha \ell \right)^{(t-h)} \times \ell^{n-t} \right)^{1/n}}{\left( \prod_{i = 1}^{h} v(\x_i) \times \left( \alpha \ell \right)^{(n-h)}\right)^{1/n}} = \alpha^{-\frac{n-t}{n}} \geq \alpha^{-\frac{1+\varepsilon}{\alpha} - \frac{1}{n}}.
	\end{align}

The $\nicefrac{-1}{n}$ term in the exponent of $\alpha$ can be neglected via a scaling argument, as follows: Construct (for analysis only) a \emph{scaled-up} instance $\I'$ consisting of $c \geq 1$ copies of the instance $\I$. For any allocation $\y$ that is $\varepsilon\text{-}\EFone$ for $\I$, the allocation $\y' = (\y,\y,\dots,\y)$ is $\varepsilon\text{-}\EFone$ for $\I'$. Write $n',\alpha',\ell',\omega'$ to denote the analogues of $n,\alpha,\ell,\omega$ in $\I'$. Also, let $\tilde{\omega}$ denote the fractional allocation $(\omega,\omega,\dots,\omega)$ in $\I'$. It is easy to see that $n' = cn$, $\alpha' = \alpha$, and $\ell' = \ell$. Moreover,
\begin{align*}
\frac{\NSW(\y)}{\NSW(\omega)} = \frac{\NSW(\y')}{\NSW(\tilde{\omega})} \geq \frac{\NSW(\y')}{\NSW(\omega')} \geq \alpha^{-\frac{1+\varepsilon}{\alpha} - \frac{1}{cn}},
\end{align*}
where the first term is for the instance $\I$, and the remaining terms are for the instance $\I'$. In addition, the relation $\NSW(\omega') \geq \NSW(\tilde{\omega})$ follows from the optimality of $\omega'$ for $\I'$. By choosing a  sufficiently large value of $c$, the term $-\nicefrac{1}{cn}$ in the exponent can be made arbitrarily small. Therefore, the lower bound in \Cref{eqn:NSWApprox_LowerBound} is (arbitrarily close to) $\alpha^{-\frac{1+\varepsilon}{\alpha}}$. Finally, notice that the function $z^{-\frac{1+\varepsilon}{z}}$ with $z\geq 0$ is minimized at $z = e$. This gives a lower bound of $e^{-(1+\varepsilon)/e}$, as desired.
\end{proof}

We will now proceed to the proof of \Cref{THM:APPROXNASH}. The proof relies on transforming a general fair division instance into one with identical valuations, and showing that the Nash social welfare of the allocation returned by \Alg{} is preserved in this transformation. 

\ApproxNash* 

\begin{proof}
For a given instance $\I = \langle [n], [m], \V \rangle$, let $\z$ and $\q$ denote the allocation and the price vector respectively that are returned by $\Alg{}$ when provided as input the $\varepsilon$-rounded version of $\I$ (the parameter $\varepsilon$ is set to a small constant). 
Let $\alpha_i \neq 0$ denote the maximum bang per buck ratio of agent $i$ with respect to $\q$. Construct a scaled instance $\I^{\text{sc}} = \langle [n], [m], \V^{\text{sc}} \rangle$ such that $v^{\text{sc}}_{i,j} = \frac{1}{\alpha_i} v_{i,j}$ for all $i$ and $j$.\footnote{A similar scaling was used by \citet{CG15approximating} in their analysis of $\NSW$ approximation.} Then, for any allocation $\y$, $\NSW(\y)$ in $\I^{\text{sc}}$ is $ \frac{1}{\left(\prod_{i=1}^n \alpha_i \right)^{1/n}}$ times $\NSW(\y)$ in the original instance $\I$. Therefore, in order to obtain the desired approximation guarantee, it suffices to show that $\z$ achieves an approximation factor of $1.45$ in the scaled instance $\I^{\text{sc}}$. 

Let $\opt$ denote a Nash optimal (integral) allocation in the original instance $\I$. 
 By the above argument, $\opt$ is Nash optimal in the scaled instance $\I^{\text{sc}}$ as well. Additionally, for each agent $i$ in $\I^{\text{sc}}$, we have that $v^{\text{sc}}_{i,j} = q_j$ for all $j \in \MBB_i$, and $v^{\text{sc}}_{i,j} < q_j $ for all $j \notin \MBB_i$. Therefore, for any agent $i$, we have $v^{\text{sc}}(\z_i) = \q(\z_i)$ (since, from \Cref{lem:ALG_outputs_eps-EF1_and_fPO}, we have that $\z_i \subseteq \MBB_i$), and $v^{\text{sc}}_i(\opt_i) \leq \q(\opt_i)$. Consequently, the Nash social welfare of the computed allocation $\z$ and the optimal allocation $\opt$ satisfy the following relations in the scaled instance $\I^{\text{sc}}$:
\begin{align}
\label{eq:nsw-alloc}
\left( \prod_{i=1}^n v^{\text{sc}}_i( \z_i) \right)^{1/n} & = \left(\prod_{i=1}^n  \q(\z_i) \right)^{1/n} 
\end{align}
and
\begin{align}
\label{ineq:nsw-opt}
\left( \prod_{i=1}^n v^{\text{sc}}_i(\opt_i) \right)^{1/n} & \leq \left(\prod_{i=1}^n  \q(\omega_i) \right)^{1/n}. 
\end{align}
We will further transform the valuations in $\I^{\text{sc}}$ to obtain an instance $\I^{\text{id}} = \langle [n], [m], \V^{\text{id}} \rangle$ with \emph{identical} valuations. Specifically, we set $v^{\text{id}}_{i,j} = q_j$ for all $i$ and $j$. We know from \Cref{lem:ALG_outputs_eps-EF1_and_fPO} that $\z$ is $3\varepsilon$-$\pEFone$ with respect to $\q$ in the original instance $\I$, and that $\z_i \subseteq \MBB_i$ for each agent $i \in [n]$. It then follows that $z$ is $3\varepsilon$-$\EFone$ in the identical valuations instance $\I^{\text{id}}$. Furthermore, when $\varepsilon=\frac{1}{300}$, the allocation $z$ is $\frac{1}{100}$--$\EFone$ in $\I^{\text{id}}$, and therefore from \Cref{lemma:id-val}, we have that 
\begin{align*}
\left(\prod_{i=1}^n  \q(\z_i) \right)^{1/n} & \geq e^{-(1+0.01)/e} \max_{\y \in \X} \left(\prod_{i=1}^n  \q(\y_i) \right)^{1/n} \\
& \geq \frac{1}{1.45} \max_{\y \in \X} \left(\prod_{i=1}^n  \q(\y_i) \right)^{1/n} \\
& \geq \frac{1}{1.45} \left(\prod_{i=1}^n  \q(\omega_i) \right)^{1/n} \\
& \geq \frac{1}{1.45} \left( \prod_{i=1}^n v^{\text{sc}}_i(\opt_i) \right)^{1/n} \qquad \text{ (using \Cref{ineq:nsw-opt})}.
\end{align*}
The previous inequality and \Cref{eq:nsw-alloc} together give us an approximation factor of $1.45$ under the valuation profile $\V^{\text{sc}}$:
\begin{align*}
\left( \prod_{i=1}^n v^{\text{sc}}_i( \z_i) \right)^{1/n} & \geq \frac{1}{1.45} \left( \prod_{i=1}^n v^{\text{sc}}_i(\opt_i) \right)^{1/n},
\end{align*}
which provides a similar approximation guarantee for the original instance $\I$. Finally, observe that the allocation $z$ can be computed in polynomial time for the above choice of $\varepsilon$ (\Cref{lem:Approx_EF_Approx_PO_Polytime}). This completes the proof of \Cref{THM:APPROXNASH}.
\end{proof}

\section{Concluding Remarks}
\label{sec:Concluding_Remarks}

We studied the problem of finding a fair and efficient allocation of indivisible goods. Our work provided a framework based on integral Fisher markets and an (approximate) price envy-freeness condition resulting in a pseudopolynomial algorithm for finding an $\EFone{}$ and $\PO$ allocation, and a polynomial time $1.45$-approximation algorithm for Nash social welfare. Determining whether there exists a (strongly) polynomial time algorithm for the problem of finding an $\EFone{}$ and $\PO$ allocation remains an interesting direction for future work. Extensions of our results to more general classes of valuations (e.g., submodular) will also be interesting.

\bibliographystyle{plainnat}
\bibliography{EF1+PO}

\newpage
\setcounter{section}{0}

\renewcommand{\thesection}{\Alph{section}}
\section{Appendix-I}
 \label{sec:Appendix_I}

\subsection{\BuildHierarchy{} subroutine}
\label{subsec:BuildHierarchy}

This section provides a polynomial time subroutine for constructing the hierarchy of agent $i$. 

Given an allocation $\x$, we use $\x^{-1}(j)$ to refer to the agent that owns the good $j$, i.e., $\x^{-1}(j) = i$ if $j \in \x_i$. Similarly, for a set of goods $G \subseteq [m]$, we write $\x^{-1}(G)$ to refer to the set of all agents that own one or more goods in the set $G$. Also, given a set $S \subseteq [n]$ of agents, we write $\MBB_S$ to refer to the set $\bigcup_{i \in S} \MBB_i$.

\begin{algorithm}[h]
 \DontPrintSemicolon
 \KwIn{An agent $i$, an allocation $\x$, and a price vector $\p$.}
 \KwOut{A hierarchy structure $\H_i = \{\H_i^{0}, \H_i^{1}, \dots  \}$ for agent $i$.}
 \BlankLine
 \tcp{Initialization}
 $\H_i^{0} \leftarrow \{i\}$\tcc*{$i$ is the Level-0 agent}
 $\ell \leftarrow 0$\;
 \tcp{Build the hierarchy level-wise}
 \While{$\H_i^{\ell}$ is non-empty}{
 $\H_i^{\ell + 1} \leftarrow \x ^{-1}\left( \MBB_{\H_i^{\ell}} \right) \setminus \{ \cup_{k=0}^\ell \H_i^{k} \}$\;
 	\tcc*{Add to $\H_i^{\ell + 1}$ any agent that is not currently in the hierarchy and is reachable from some member of $\H_i^{\ell}$ via an $\MBB{}$-allocation edge.}
 	$\ell \leftarrow \ell + 1$\;
 }
 \KwRet{$\H_i = \{\H_i^{0}, \H_i^{1},\dots\}$}.\;
 \caption{\BuildHierarchy}
 \label{alg:Hierarchy_Construction}
\end{algorithm}

\subsection{Proof of Lemma~\ref{lem:ALG_RunningTime_PowersOfr}}
\label{subsec:Proof_Of_ALG_RunningTime_PowersOfr}

This section provides the proof of \Cref{lem:ALG_RunningTime_PowersOfr}. Throughout this section, we will use $i_t \in [n]$ to denote the least spender at time step $t$. Additionally, we will use $\x^t$ and $\p^t$ to denote the allocation and the price vector respectively at time step $t$.

\ALGRunningTime*
\begin{proof}
Follows from \Cref{lem:ALG_Phase2_RunningTimeBound,lem:ALG_Phase3_RunningTimeBound}, which are stated below.
\end{proof}

\begin{restatable}{lemma}{ALGPhaseTwoRunningTimeBound}
 \label{lem:ALG_Phase2_RunningTimeBound}
 Phase 2 of \Alg{} can continue for at most $\poly(n, m, \frac{1}{\varepsilon}) \cdot \ln m v_{\max}$ consecutive time steps before a Phase 3 step occurs.
\end{restatable}

\begin{restatable}{lemma}{ALGPhaseThreeRunningTimeBound}
 \label{lem:ALG_Phase3_RunningTimeBound}
 \Alg{} can perform at most $n \log_{(1+\varepsilon)} m v_{\max}$ number of Phase 3 steps.
\end{restatable}

The proofs of \Cref{lem:ALG_Phase2_RunningTimeBound,lem:ALG_Phase3_RunningTimeBound} rely on several intermediate results (\Cref{lem:Least_Spender's_Spending_Never_Decreases,lem:Cap_on_least-spender's_spending,lem:ALG_Phase2_Swaps_PolynomialBound,lem:ALG_Phase2_IdentityChangeImplication,lem:ALG_Phase2_NumberOfIdentityChanges}).

\begin{restatable}{lemma}{LeastSpenderSpendingNeverDecreases}
 \label{lem:Least_Spender's_Spending_Never_Decreases}
 The spending of the least spender cannot decrease with time, i.e., for each time step $t$, $\p^t(\x^t_{i_t}) \leq \p^{t+1}(\x^{t+1}_{i_{t+1}})$.
 \end{restatable}
 \begin{proof}
 There are exactly three ways in which the spending of the least spender can change between consecutive time steps: (1) due to a swap operation in Phase 2, (2) a price-rise in Phase 3, or (3) an identity change. In Phase 2, the spending of the least spender can be affected via a swap operation only if it receives a good, which results in an increase in its spending. In that case, the agent losing the good cannot become the new least spender due to the $\varepsilon$-path-violator condition. Similarly, in Phase 3, the spending of the least spender cannot decrease since the prices of the goods can only increase. Finally, any identity change, either in Phase 2 or in Phase 3, occurs only when the spending of the old least spender grows beyond that of the new one, once again implying the stated condition.
\end{proof}

\begin{restatable}{lemma}{CapOnLeastSpendersSpending}
 \label{lem:Cap_on_least-spender's_spending}
 Let $t$ denote a Phase 3 step. Then, $\p^t(\x^t_{i_t}) \leq m v_{\max}$.
\end{restatable}

\begin{restatable}{lemma}{PhaseTwoSwapPolyBound}
 \label{lem:ALG_Phase2_Swaps_PolynomialBound}
 \Alg{} can perform at most $\poly(n,m)$ number of consecutive swap operations before either the identity of the least spender changes or a Phase 3 step occurs.
\end{restatable}

The proofs of \Cref{lem:Cap_on_least-spender's_spending,lem:ALG_Phase2_Swaps_PolynomialBound} rely on a number of intermediate results, and we defer them to \Cref{subsubsec:MissingProof_Cap_on_least-spender's_spending,subsubsec:Proof_Of_ALG_Phase2_Swaps_PolynomialBound} respectively.

\begin{lemma}
 \label{lem:ALG_Phase2_IdentityChangeImplication}
 Consider a series of consecutive time steps consisting entirely of Phase 2 operations, i.e., either swap operations or change in the identity of the least spender. Let $t$ be a time step at which an agent $i$ ceases to be the least spender, and let $t' > t$ be the first time step after $t$ at which $i$ once again becomes the least spender. Let $(\x,\p)$ and $(\x',\p')$ denote the corresponding allocation and price vectors. Then, either $\x_i \subsetneq \x'_i$ or $\p'(\x'_i) \geq (1+\varepsilon) \p(\x_i)$, or both.
\end{lemma}
 \begin{proof}
 Observe that a change in the identity of the least spender during Phase 2 is always preceded by the previous least spender receiving a good via a swap operation. This means that agent $i$ must receive a good at time $t$. If, in addition, agent $i$ does not lose any good during the time interval between $t$ and $t'$, then we already have that $\x_i \subsetneq \x'_i$ and the claim follows. Therefore, for the rest of the proof, we will assume that agent $i$ loses one or more goods between $t$ and $t'$.
 
 Among all time steps between $t$ and $t'$ at which agent $i$ loses a good, let $\overline{t}$ denote the last one. Let $j \in [m]$ denote the good lost by agent $i$ at time $\overline{t}$, and let $k$ denote the least spender at that time. Also, let $(\overline{\x},\overline{\p})$ denote the allocation and price vector just before $i$ loses the good $j$.
 
  \Cref{lem:Least_Spender's_Spending_Never_Decreases} that the spending of the least spender cannot decrease with time. Thus,
 \begin{equation}
 	\overline{\p}(\overline{\x}_k) \geq \p(\x_i).
 	\label{eqn:ALG_Phase2_temp1}
 \end{equation}
 
 Since agent $i$ loses the good $j$ at time $\overline{t}$, it must be an $\varepsilon$-path-violator with respect to $(\overline{\x},\overline{\p})$. Hence,
 \begin{equation}
 	\overline{\p}(\overline{\x}_i \setminus \{j\}) > (1+\varepsilon) \overline{\p}(\overline{\x}_k).
 	\label{eqn:ALG_Phase2_temp2}
 \end{equation}
 
 Finally, since $j$ is the last good lost by $i$ before the time step $t'$, we have that
 \begin{equation}
 	\p'(\x'_i) \geq \overline{\p}(\overline{\x}_i \setminus \{j\}).
 	\label{eqn:ALG_Phase2_temp3}
 \end{equation}
 
 \Cref{eqn:ALG_Phase2_temp1,eqn:ALG_Phase2_temp2,eqn:ALG_Phase2_temp3} together give us the desired result.
\end{proof}

\begin{lemma}
 \label{lem:ALG_Phase2_NumberOfIdentityChanges}
The identity of the least spender can change at most $\poly(n, m, \frac{1}{\varepsilon}) \cdot \ln m v_{\max}$ number of times during Phase 2 before a Phase 3 step occurs.
\end{lemma}
 \begin{proof}
  Recall from \Cref{lem:ALG_Phase2_IdentityChangeImplication} that each time \Alg{} cycles back to an agent $i$ as the least spender, either the allocation of $i$ strictly grows by at least one good, or its spending grows at least by a multiplicative factor of $(1+\varepsilon)$. By pigeonhole principle, for every $(n+1)$ identity change events, \Alg{} must cycle back to some (fixed) agent. Therefore, for every $(n+1)$ consecutive identity change events (possibly interspersed with swap operations), either that agent obtains an extra good (without losing any) or its spending grows by a factor of $(1+\varepsilon)$. This observation, along with the fact that the spending of the least spender can never decrease with time (\Cref{lem:Least_Spender's_Spending_Never_Decreases}), implies that for every $m(n+1)$ identity changes, the spending of the least spender must increase by a factor of $(1+\varepsilon)$. Furthermore, we know from \Cref{lem:Cap_on_least-spender's_spending} that the spending of the least spender at the beginning of each Phase 3 step is at most $m v_{\max}$. Hence, assuming that the initial spending of the least spender is at least $1$ (refer to \Cref{subsec:CornerCases} for explanation of why this assumption is without loss of generality), there can be at most $\poly(m,n) \cdot \log_{(1+\varepsilon)} m v_{\max}$ identity changes during Phase 2 before a Phase 3 step occurs. By using $\ln (1+\varepsilon) \geq \varepsilon - \varepsilon^2$, we obtain the desired result.
\end{proof}

\ALGPhaseTwoRunningTimeBound*
 \begin{proof}
Follows from \Cref{lem:ALG_Phase2_Swaps_PolynomialBound,lem:ALG_Phase2_NumberOfIdentityChanges}.
\end{proof}

\ALGPhaseThreeRunningTimeBound*
 \begin{proof}
Recall that each Phase 3 step increases the prices of the goods owned by the members of the hierarchy by a multiplicative factor of $\alpha$, where $\alpha = \min\{\alpha_1,\alpha_2,\alpha_3\}$. Since the algorithm terminates if $\alpha = \alpha_2$, we will assume for the rest of the proof that $\alpha$ is equal to either $\alpha_1$ or $\alpha_3$.

\emph{Price-rise by $\alpha = \alpha_1$}: In order to bound the number of price-rise steps with $\alpha = \alpha_1$, we will show that each such step multiplicatively increases the spending of the least spender by a positive integral power of $(1+\varepsilon)$. Then, by using \Cref{lem:Least_Spender's_Spending_Never_Decreases,lem:Cap_on_least-spender's_spending}, we can conclude that there can be at most $\log_{(1+\varepsilon)} m v_{\max}$ price-rise steps with $\alpha = \alpha_1$.\footnote{We are assuming here that the initial spending of the least spender is at least $1$; refer to \Cref{subsec:CornerCases} for an explanation of why this assumption is without loss of generality.}

 Notice that since the valuations are power-of-$(1+\varepsilon)$, the initial prices of all the goods at the end of Phase 1 are also power-of-$(1+\varepsilon)$. We claim that if all prices are power-of-$(1+\varepsilon)$ prior to the price-rise step with $\alpha = \alpha_1$, then the same continues to hold after the price-rise. Indeed, if all prices are power-of-$(1+\varepsilon)$, then all bang per buck ratios must also be (possibly negative) integral powers of $(1+\varepsilon)$. Moreover, $\alpha_1$ is itself a ratio of two such bang per buck ratios, and thus it must also be an integral power of $(1+\varepsilon)$. Each price-rise by $\alpha = \alpha_1$ multiplicatively increments the prices by a power of $(1+\varepsilon)$. Furthermore, this increment must be strict because the price-rise step introduces a new $\MBB{}$ edge between a member of the hierarchy and a good outside the hierarchy. Hence, the spending of the least spender grows by a positive integral power of $(1+\varepsilon)$.

\emph{Price-rise by $\alpha = \alpha_3$}: We will now provide a bound on the number of price-rise steps with $\alpha = \alpha_3$. While the price-rise in this case is also in integral powers of $(1+\varepsilon)$ (by definition of $\alpha_3$), the spending of the least spender may not increase by the same factor after \emph{each} such step due to a possible change in the identity of the least spender. However, by pigeonhole principle, the algorithm must cycle back to the same agent after $n$ steps, and therefore the spending of the least spender must grow multiplicatively (at least) by a positive integral power of $(1+\varepsilon)$ after every $n$ price-rise steps with $\alpha = \alpha_3$. By a similar reasoning as for the case $\alpha = \alpha_1$, we get the desired result.
\end{proof}

\subsubsection{Proof of Lemma~\ref{lem:Cap_on_least-spender's_spending}}
\label{subsubsec:MissingProof_Cap_on_least-spender's_spending}

The proof of \Cref{lem:Cap_on_least-spender's_spending} relies on a series of intermediate results (\Cref{lem:EnvySetShrinks,lem:pEF1envied_agents_do_not_receive_goods,lem:Hierarchy__pEF1free_after_Phase2,lem:Adversary's_Spending_Never_Increases}), which we state and prove below. We will use $E_t \subset [n]$ to denote the set of all $3\varepsilon$-violators at time $t$. That is, $$E_t \coloneqq \left\{ h \in [n] : (1+3\varepsilon) \p^t(\x^t_{i_t}) < \p^t(\x^t_h \setminus \{j\}) \text{ for every good } j \in \x^t_h \right\}.$$

\begin{restatable}{lemma}{EnvySetShrinks}
 \label{lem:EnvySetShrinks}
 Let $t$ and $t'$ be two time steps at which \Alg{} performs a price-rise in Phase 3 such that $t < t'$. Then, $E_{t'} \subseteq E_t$.
\end{restatable}
\begin{proof}
It suffices to prove \Cref{lem:EnvySetShrinks} for consecutive price-rise steps $t$ and $t'$ (possibly including Phase 2 events between them). Suppose, for contradiction, that there exists an agent $k \in E_{t'} \setminus E_t$. Our proof consists of two main arguments: First, we will show that $k$ cannot turn into a $3\varepsilon$-violator due to the price-rise at time $t$. (This would imply that the only way $k$ can turn into a $3\varepsilon$-violator is via a swap operation.) Second, we will show that if there is a swap operation at time $\overline{t}$ (for some $t < \overline{t} < t'$) that turns $k$ into a $3\varepsilon$-violator, then there is a subsequent swap operation at $(\overline{t} + 1)$ that turns it back into a non-$\varepsilon$-violator. This will contradict the fact that $k$ is a $3\varepsilon$-violator at the beginning of the price-rise step at $t'$.

We will start by showing that $k$ cannot turn into a $3\varepsilon$-violator due to the price-rise at time $t$. We perform case analysis for whether or not $k \in \H_{i_t}$. To begin with, suppose that $k \in \H_{i_t}$. Then, $k$ cannot be an $\varepsilon$-violator \emph{before} the price-rise at time $t$ (otherwise it would also be an $\varepsilon$-path-violator, and $\Alg{}$ would continue with Phase 2). Thus,
\begin{equation*}
	(1+\varepsilon) \p^t(\x^t_{i_t}) \geq \p^t(\x^t_k \setminus \{j\}) \text{ for some } j \in \x^t_k.
\end{equation*}

A similar condition continues to hold \emph{after} the price-rise, since prices are always raised uniformly.
\begin{equation}
\label{eqn:No_eps_violator_after_price_rise}
	(1+\varepsilon) \p^{t+1}(\x^{t+1}_{i_t}) \geq \p^{t+1}(\x^{t+1}_k \setminus \{j\}) \text{ for some } j \in \x^{t+1}_k.
\end{equation}

Therefore, at time $(t+1)$, agent $k$ cannot be an $\varepsilon$-violator with respect to any agent in $\H_{i_t}$. It is, however, possible that $k$ is an $\varepsilon$-violator at time $(t+1)$ with respect to some agent \emph{outside} $\H_{i_t}$. Nevertheless, we will show that $k$ cannot be a $3\varepsilon$-violator. Specifically, let $h$ be the least spender outside $\H_{i_t}$ at time $t$, i.e., $h \in \arg\min_{a \in [n] \setminus \H_{i_t}} \p^t(\x^t_a)$. Recall that the condition $\alpha \leq \alpha_3$ in Line~\ref{algline:alpha_defn} of \Alg{} implies that
\begin{equation*}
	(1+\varepsilon) \p^{t+1}(\x^{t+1}_h) \geq \p^{t+1}(\x^{t+1}_{i_t}).
\end{equation*}

Along with \Cref{eqn:No_eps_violator_after_price_rise}, this gives
\begin{equation*}
	(1+\varepsilon)^2 \p^{t+1}(\x^{t+1}_h) \geq \p^{t+1}(\x^{t+1}_k \setminus \{j\}) \text{ for some } j \in \x^{t+1}_k.
\end{equation*}

Since $\varepsilon < 1$, we have that $(1+\varepsilon)^2 < 1+3\varepsilon$, which implies that $k$ is not a $3\varepsilon$-violator at time $(t+1)$ with respect to any agent.

Now suppose that $k \notin \H_{i_t}$. Since $k \notin E_t$ by assumption, and the spending of the agents outside the hierarchy remains unaffected due to the price-rise, we once again get that $k \notin E_{t+1}$. This proves that $k$ cannot turn into a $3\varepsilon$-violator due to the price-rise at time $t$.

We will now proceed to show that if there is a swap operation at time $\overline{t}$ (for some $t < \overline{t} < t'$) that turns $k$ into a $3\varepsilon$-violator, then there is a subsequent swap operation at $(\overline{t} + 1)$ that turns it back into a non-$\varepsilon$-violator. Suppose that $k$ (at level $\ell$ in the hierarchy) becomes a $3\varepsilon$-violator after receiving a good $j$ via a swap at time step $\overline{t}$. Recall that a swap operation involves transferring a good from an agent at a higher level to another agent at a lower level in the hierarchy. Furthermore, \Alg{} performs a swap for an agent at level $(\ell+1)$ only if no agent in the levels $1,2,\dots,\ell$ is an $\varepsilon$-path violator. Therefore, $k$ cannot be an $\varepsilon$-path violator before the swap, i.e., there exists a good $j'$ on an alternating path of length $2\ell$ from $i_{\overline{t}}$ to $k$ such that
\begin{equation*}
	(1+\varepsilon) \p^{\overline{t}}(\x^{\overline{t}}_{i_{\overline{t}}}) \geq \p^{\overline{t}}(\x^{\overline{t}}_k \setminus \{j'\}).
\end{equation*}
Moreover, since $k$ becomes a $3\varepsilon$-violator (and hence, an $\varepsilon$-path violator) after receiving the good $j$, we have that
\begin{equation}
\label{eqn:temp2_EnvySetShrinks}
	(1+3\varepsilon) \p^{\overline{t}+1}(\x^{\overline{t}+1}_{i_{\overline{t}+1}}) < \p^{\overline{t}+1}(\x^{\overline{t}}_k \cup \{j\} \setminus \{j'\}).
\end{equation}
Since neither the identity (or allocation) of the least spender nor the price-vector changes in this process, \Cref{eqn:temp2_EnvySetShrinks} can be rewritten as
\begin{align}
\label{eqn:temp3_EnvySetShrinks}
	(1+3\varepsilon) \p^{\overline{t}}(\x^{\overline{t}}_{i_{\overline{t}}}) < \p^{\overline{t}}(\x^{\overline{t}}_k \cup \{j\} \setminus \{j'\}).
\end{align}
Notice that the swap involving the good $j$ does not affect the alternating path from $i_{\overline{t}}$ to $k$ via the good $j'$, and therefore $k$ continues to be at level $\ell$. In fact, $k$ is the only agent on level $\ell$ or below that is an $\varepsilon$-path-violator. Therefore, in a subsequent swap operation, the good $j'$ will be taken away from $k$, resulting in the allocation $(\x^{\overline{t}}_k \cup \{j\} \setminus \{j'\})$ for $k$. After this step, agent $k$ once again becomes a non-$\varepsilon$-violator with respect to the good $j$, providing the desired contradiction.
\end{proof}

\begin{restatable}{lemma}{EnviedAgentsDoNotGetGoods}
 \label{lem:pEF1envied_agents_do_not_receive_goods}
  Let $t$ and $t'$ be two time steps at which \Alg{} performs a price-rise in Phase 3 such that $t < t'$. Then, for any agent $k \in E_{t'}$, $\x^{t'}_k \subseteq \x^t_k$.
\end{restatable}
 \begin{proof}(Sketch)
 The proof is very similar to that of \Cref{lem:EnvySetShrinks}. Suppose, for contradiction, that there exists a good $j \in \x^{t'}_k \setminus \x^t_k$ for some agent $k \in E_{t'}$. Then, agent $k$ must have acquired the good $j$ via a swap operation at time $\overline{t}$ (between $t$ and $t'$). This means that agent $k$ cannot be an $\varepsilon$-path violator at time $\overline{t}$, and thus cannot be a $3\varepsilon$-violator. By an argument similar to \Cref{lem:EnvySetShrinks}, we can argue that $k$ cannot be a $3\varepsilon$-violator at any subsequent price-rise event, contradicting $k \in E_{t'}$.
\end{proof}

\begin{restatable}{lemma}{HierarchyPriceEnvyFreeAfterPhaseTwo}
 \label{lem:Hierarchy__pEF1free_after_Phase2}
  At the beginning of each price-rise step at time $t$, $E_t \cap \H_{i_t} = \emptyset$.
\end{restatable}
\begin{proof}
Suppose, for contradiction, that there exists an agent $k \in E_t \cap \H_{i_t}$. Since $k \in E_t$, we have $(1+3\varepsilon) \p^t(\x^t_{i_t}) < \p^t(\x^t_k \setminus \{j\})$ for every good $j \in \x^t_k$. Furthermore, since $k \in \H_{i_t}$, there must exist an alternating path from the least spender $i_t$ to $k$ that involves some good $j' \in \x^t_k$. Thus, $k$ is also an $\varepsilon$-path violator, which means that \Alg{} will perform a swap operation in Phase 2 at time $t$, as opposed to a price-rise operation.
\end{proof}

\begin{restatable}{lemma}{SpendingOfAdversaryCanNeverIncrease}
 \label{lem:Adversary's_Spending_Never_Increases}
 Let $t$ and $t'$ be two time steps at which \Alg{} performs a price-rise such that $t < t'$. Then, the spending of any agent $k \in E_{t'}$ at time $t'$ is at most that at time $t$, i.e., $\p^{t'}(\x^{t'}_k) \leq \p^t(\x^t_k)$.
\end{restatable}
 \begin{proof}
 Assume, without loss of generality, that $t$ and $t'$ correspond to consecutive price-rise steps (possibly including Phase 2 events between them). From \Cref{lem:pEF1envied_agents_do_not_receive_goods}, we have that $\x^{t'}_k \subseteq \x^t_k$. Therefore, it suffices to show that $\p^{t'}(\x^{t'}_k) = \p^t(\x^{t'}_k)$, i.e., the prices of the goods in the set $\x^{t'}_k$ do not vary between $t$ and $t'$. We know from \Cref{lem:EnvySetShrinks} that $E_t \supseteq E_{t'}$, and hence $k \in E_{t}$. \Cref{lem:Hierarchy__pEF1free_after_Phase2} then implies that $k \notin \H_{i_t}$, which means that the price-rise step at time $t$ does not affect the prices of the goods owned by $k$ at time $t$, namely $\x^t_k$. Since $\x^{t'}_k \subseteq \x^t_k$, the same holds for the goods in $\x^{t'}_k$. The lemma now follows since, by assumption, there is no other price-rise step between $t$ and $t'$.
\end{proof}

We will finish this section with a proof of \Cref{lem:Cap_on_least-spender's_spending}.

\CapOnLeastSpendersSpending*
 \begin{proof}
Notice that the set $E_t$ must be non-empty at the beginning of each price-rise step, otherwise $(\x,\p)$ is already $3\varepsilon$-$\pEFone{}$ and \Alg{} will terminate before entering Phase 3. By definition of the set $E_t$, the spending of the least spender $i_t$ is at most the spending of any agent $h \in E_t$, i.e., $\p^t(\x^t_{i_t}) \leq \p^t(\x^t_{h})$ for every $h \in E_t$. From \Cref{lem:Adversary's_Spending_Never_Increases}, we know that $\p^t(\x^t_h) \leq \p^{0}(\x^{0}_h)$, where the superscript ``$0$'' denotes the first price-rise event. The spending of agent $h$ just before the first price-rise event, namely $\p^{0}(\x^{0}_h)$, is at most the sum of prices of all goods right after Phase 1 (this is because Phase 2 does not affect the prices of the goods). The price of any good at the end of Phase 1 is simply the highest valuation of an agent for that good. Therefore, $\p^t(\x^t_{i_t}) \leq \p^t(\x^t_{h}) \leq m v_{\max}$.
\end{proof}

\subsubsection{Proof of Lemma~\ref{lem:ALG_Phase2_Swaps_PolynomialBound}}
\label{subsubsec:Proof_Of_ALG_Phase2_Swaps_PolynomialBound}

Let $\H_{i_t}$ denote the hierarchy at time step $t$. The \emph{level} of an agent $h \in [n]$ at time $t$ is defined as
\[
\level(h,t) \coloneqq
\begin{cases}
    \ell  & \quad \text{if } h \in \H_{i_t}^\ell, \text{ and}\\
    n     & \quad \text{if } h \notin \H_{i_t}.
  \end{cases}
\]
Notice that for any agent $h \in \H_{i_t}$, $\level(h,t) \leq n-1$, since there are $n$ agents overall, and the least spender $i_t$ is assumed to be at level $0$.

We say that a good $j \in [m]$ is \emph{critical} for an agent $h \in \H_{i_t}^\ell \setminus \{i_t\}$ at time $t$ if $j \in \x^t_h$, and there is an alternating path of length $2\ell$ from the least spender $i_t$ to $h$ that includes the good $j$. We denote the set of all critical goods for an agent $h$ at time $t$ by $G_{h,t}$. It is easy to see that the set $G_{h,t}$ is non-empty if and only if $h \in \H_{i_t}^\ell \setminus \{i_t\}$ for some $\ell \in \{1,2,\dots\}$.

\PhaseTwoSwapPolyBound*

 \begin{proof}
 We will show via a potential function argument that each swap operation causes a drop of at least $1$ in the value of a non-negative function that is bounded above by $\poly(m,n)$, implying that \Alg{} can perform at most $\poly(m,n)$ such swaps.

Suppose that \Alg{} performs a swap operation at time step $t$ that involves transferring a good $j \in [m]$ from an agent $h_\ell$ at level $\ell$ to another agent $h_{\ell-1}$ at level $\ell - 1$. Consider the function $f(t)$ defined as
$$f(t) \coloneqq \sum\limits_{h \in [n] \setminus \{i_t\}} m \left( n - \level(h,t) \right) + |G_{h,t}|.$$
Notice that $f(t)$ is always non-negative, and its maximum value is at most $\poly(n,m)$, since $\level(h,t) \geq 0$ and $|G_{h,t}| \leq m$ for every agent $h \in [n] \setminus \{i_t\}$. The rest of the proof will separately analyze the contribution of the agents $h_\ell$, $h_{\ell - 1}$, and any other agent $h \in [n] \setminus \{i_t,h_{\ell-1},h_\ell\}$ to the term $f(t)$.

First, consider the agent $h_\ell$. We will argue that either $\level(h_\ell,t+1) = \level(h_\ell,t)$ and $G_{h_\ell,t+1} \subsetneq G_{h_\ell,t}$, or $\level(h_\ell,t+1) > \level(h_\ell,t)$. Notice that both conditions cause the potential term to drop by at least $1$ (all else being constant): The first condition simply causes a drop of at least $1$ in the term $|G_{h,t}|$ without a change in the level term; the second term causes a drop of at least $1$ in the level term, which is scaled by a factor of $m$, and thus cannot be offset by any possible increase in the critical goods term. The reason why one of these conditions always holds is as follows: Notice that $j$ is a critical good for $h_\ell$ before the swap takes place, i.e., $j \in G_{h,t}$. If $j$ is the only good in the set $G_{h,t}$, then, after the swap operation, either $h_\ell$ moves to a higher level (than $\ell$) in the hierarchy, or it ceases to be a member of the hierarchy. Either way, we have that $\level(h_\ell,t+1) > \level(h_\ell,t)$. On the other hand, if $j$ is not the only good in the set $G_{h,t}$, then the level of $h_\ell$ does not change upon the loss of the good $j$. 
 However, since $h_\ell $ does not receive any good during the swap operation, and none of the previously non-critical goods become critical, we have that $G_{h_\ell,t+1} \subsetneq G_{h_\ell,t}$.
 
Next, consider the agent $h_{\ell - 1}$. We will show that $\level(h_{\ell - 1},t+1) = \level(h_{\ell - 1},t)$ and $G_{h_{\ell - 1},t+1} = G_{h_{\ell - 1},t}$, implying that all else being constant, the potential cannot increase. Notice that a swap operation can (possibly) create one or more new alternating paths from the least spender $i$ to the agent $h_{\ell-1}$. We will argue that any such path must be of length at least $2\ell$, and therefore it can neither lower the level nor make $j$ a critical good for $h_{\ell-1}$. Indeed, any new path to $h_{\ell-1}$ must involve the good $j$ and another agent $h'$ such that $j \in \MBB_{h'}$. Since $j$ was previously owned by $h_\ell$, there must already exist an alternating path from $i$ to $h_\ell$ via $h'$ \emph{before} the swap takes place. Since the level of $h_\ell$ before the swap is $\ell$, any such path must be of length at least $2(\ell-1)$, implying that $\level(h',t) \geq \ell - 1$. This, in turn, means that any new path to $h_{\ell-1}$ (via $h'$) must be of length at least $2\ell$, which is strictly longer than the existing path of length $2(\ell-1)$ that defines its level.
 
Finally, we will show that the contribution of any agent $h \in [n] \setminus \{i,h_{\ell-1},h_\ell\}$ to the potential $f$ does not increase due to the swap operation. To begin with, note that $\level(h,t+1) \geq \level(h,t)$: the only way in which the level of agent $h$ can decrease is via a newly created alternating path (at time $t+1$) from $i$ to $h$ that involves the good $j$ (this is because any path from $i$ to $h$ that does not involve the good $j$ must also have existed prior to the swap operation). Any such path must consist of two parts: (1) an alternating path from $i$ to agent $h_{\ell-1}$ and (2) an alternating path $P$ from $h_{\ell-1}$ to $h$. Since the level of $h_{\ell-1}$ remains unchanged and all the edges in $P$ were present before the swap operation, we get that $\level(h,t+1) \geq \level(h,t)$.

An analogous argument shows that the distance (i.e., the length of the shortest alternating path) of any good $j' \neq j$ from $i$ cannot decrease after the swap operation. This observation implies that the containment $G_{h,t+1} \subseteq G_{h,t}$ does not hold only if the level of $h$ strictly increases: consider a $j' \in G_{h,t+1} \setminus G_{h,t}$ (if one exists) and note that $j'$ is allocated to $h$ even before the swap (the allocation of agent $h$ does not change at time $t$). Since $j'$ was not a critical good for $h$ before the swap, its distance from $i$ must have be strictly greater than $\level(h,t)$. The distance of $j'$ from $i$ defines the level of agent $h$ at time $t+1$, hence we get that $\level(h,t+1) > \level(h,t)$. In the potential $f$, for any agent $h$ an increase in the level compensates for any possible change in the cardinality of the critical set of goods. Hence, the contribution of any agent $h \in [n] \setminus \{i,h_{\ell-1},h_\ell\} $ to $f$ cannot increase at time $t$. 

Thus, we have shown that for each agent $h \in [n] \setminus \{i\}$, the contribution to the potential does not increase and, in fact, for agent $h_\ell$ it strictly decreases. This forces the potential function to drop by at least $1$ after each swap operation, giving us the desired result.
\end{proof}

\subsection{Proof of \Cref{lem:Small_delta_PO}}
\label{subsec:OmittedProofs_ALG_analysis_generalcase}


Let $\p$ denote the price vector returned by $\Alg{}$ on the input instance $\I'$. Our first result in this section (\Cref{lem:Bound_on_final_price_of_a_fixed_good}) provides a bound on the final price of a good $j \in [m]$. 

\begin{restatable}{lemma}{FixedGoodPriceUpperBound}
\label{lem:Bound_on_final_price_of_a_fixed_good}
For any good $j \in [m]$, $p_j \leq m^2 v_{\max}^3$, where $p_j$ is the price of good $j$ upon termination of \Alg{}.
\end{restatable}
\begin{proof}
Let $\{t_1,t_2,\dots,t_\ell\}$ denote the set of all time steps (during the execution of $\Alg{}$) at which the price of good $j$ increases, and let  $\{ \alpha_{t_1},\alpha_{t_2},\dots,\alpha_{t_\ell} \}$ denote the corresponding set of multiplicative price jumps. The set of least spenders at these time steps is given by $\{i_{t_1},i_{t_2},\dots,i_{t_\ell}\}$. Let $s(t) \coloneqq \p^t(\x^t_{i_t})$ denote the spending of the least spender at time step $t$.\footnote{Recall from \Cref{subsec:ALG_analysis_powersofr} that spending `at time step $t$' refers to the spending \emph{before} the event at time step $t$ takes place.} Our proof relies on the following two claims:
\begin{enumerate}[\text{Claim} 1:]
	\item $s(t_2) \geq \alpha_{t_1} s(t_1)$, $s(t_3) \geq \alpha_{t_2} s(t_2)$, $\dots$, $s(t_\ell) \geq \alpha_{t_{\ell-1}} s(t_{\ell-1})$,  and
	\item $\alpha_{t_\ell} \leq m v_{\max}$.
\end{enumerate}
Claim 1 shows that the spending of the least spender (say, $s (t_{i+1})$) \emph{before} a price-rise involving the good $j$ is at least that of the previous least spender (i.e., the agent that was the least spender for the previous price-rise involving the good $j$) \emph{after} the corresponding price-rise (i.e., $\alpha_{t_i} s(t_i)$).\footnote{Recall that in a price-rise step involving good $j$, the spending of the least spender grows by the same multiplicative factor as the price of good $j$.} Claim 2 provides a bound on the final price-rise involving the good $j$. Before proving these claims, we will describe how they lead to the desired relation $p_j \leq m^2 v_{\max}^3$.

First, observe that Claim 1 implies that $s(t_\ell) \geq \alpha_{t_1} \alpha_{t_2} \ldots \alpha_{t_{\ell-1}} s(t_1)$. Using \Cref{lem:Cap_on_least-spender's_spending} for the time step $t_\ell$, we have that $s(t_\ell) \leq m v_{\max}$. This implies that $\alpha_{t_1} \alpha_{t_2} \ldots \alpha_{t_{\ell-1}} \leq m v_{\max}$, since the \emph{initial} spending of each agent (i.e., spending at the end of Phase 1) is assumed to be at least $1$ (\Cref{subsec:CornerCases}), and the spending of the least spender cannot decrease with time (\Cref{lem:Least_Spender's_Spending_Never_Decreases}). Along with Claim 2, this gives $\alpha_{t_1} \alpha_{t_2} \ldots \alpha_{t_{\ell-1}} \alpha_{t_{\ell}} \leq m^2 v^2_{\max}$. The final price of good $j$ is given by $p_j = p^{0}_j \alpha_{t_1} \alpha_{t_2} \ldots \alpha_{t_{\ell-1}} \alpha_{t_\ell}$, where $p^{0}_j$ denotes the price of good $j$ at the end of Phase 1. The desired bound on $p_j$ follows by observing that the initial price of any good is at most $v_{\max}$.

We will first prove Claim 2. Recall that the price-rise factor $\alpha$ in \Alg{} is chosen as $\alpha =\min\{\alpha_1,\alpha_2,\alpha_3\}$; thus, in particular, $\alpha \leq \alpha_2$. Therefore, in order to prove a bound on $\alpha_\ell$, we can assume without loss of generality that $\alpha = \alpha_2$. Under this assumption, the price-rise step in this case involves raising the spending of the least spender $i_t$ until the allocation becomes $3\varepsilon$-$\pEFone$ (or equivalently, until the set $E_t$ of the $\varepsilon$-violators at time step $t$ becomes empty). Using arguments similar to those in the proof of \Cref{lem:Cap_on_least-spender's_spending}, we can show that spending of the highest spender in $E_t$ can be at most its initial spending (i.e., spending at the end of Phase 1), hence at most $m v_{\max}$. Thus, the price-rise factor $\alpha_\ell$ is also at most $m v_{\max}$. This proves Claim 2.

We will now prove Claim 1 for the time-steps $t_1$ and $t_2$ (that is, we will show that $s(t_2) \geq \alpha_{t_1} s(t_1)$); the analysis for other time-steps follows analogously. Note that if the identity of the least-spender does not change after a price-rise at $t_1$, then we have that $s(t_1+1) = \alpha_{t_1} s(t_1)$. Additionally, since the spending of the least spender is non-decreasing with time (\Cref{lem:Least_Spender's_Spending_Never_Decreases}), we have that $s(t_2) \geq s(t_1+1) = \alpha_{t_1} s(t_1)$, and the claim follows. Therefore, we will assume, without loss of generality, that the identity of the least-spender necessarily changes after the price-rise at $t_1$.

For ease of presentation, we will use $[t_a,t_b] \coloneqq \{t_a, t_a + 1, \dots, t_b - 1, t_b\}$ and $[t_a,t_b) \coloneqq \{t_a, t_a + 1, \dots, t_b - 1\}$ to denote the set of \emph{all} time-steps (both Phase 2 and Phase 3) between $t_a$ and $t_b$ including and excluding $t_b$ respectively.\footnote{Here, $t_a$ and $t_b$ are \emph{any} two time-steps and do not necessarily correspond to price-rise steps involving good $j$.} In addition, we will say that an agent $k \in [n]$ \emph{experiences price-rise at time $t$} if $k$ belongs to the hierarchy during the price-rise at time-step $t$, i.e., $k \in \H_{i_t}$. Similarly, we will say that agent $k$ experiences price-rise \emph{during} $[t_a,t_b]$ (respectively, $[t_a,t_b)$) if $k$ experiences price-rise for some $t \in [t_a,t_b]$ (respectively, $t \in [t_a,t_b)$).

Let $\tau \in [t_1,t_2]$ be a time-step (either Phase 2 or Phase 3) such that
\begin{enumerate}
	\item the least-spender at $\tau$, namely $i_\tau$, experiences price-rise during $[t_1,\tau)$, and
	\item there does not exist $\hat{\tau} \in [t_1,t_2]$ with $\hat{\tau} < \tau$ such that the least-spender at $\hat{\tau}$, namely $i_{\hat{\tau}}$, experiences price-rise during $[t_1,\hat{\tau})$.
\end{enumerate}

Among all the time-steps in $[t_1,\tau)$ at which the agent $i_\tau$ (as defined above) experiences price-rise, let $\tau'$ be the last one. Our proof relies on the following three observations:
\begin{enumerate}[\text{Fact }I:]
	\item There exists $\tau \in [t_1,t_2]$ that satisfies condition (1).
	\item The spending of $i_\tau$ at time-step $\tau'+1$ is at least $\alpha_{t_1} s(t_1)$, i.e., $\p^{\tau'+1}(\x^{\tau'+1}_{i_\tau}) \geq \alpha_{t_1} s(t_1)$.
	\item The spending of $i_\tau$ at $\tau$ is at least that at $\tau'+1$, i.e., $\p^\tau(\x^\tau_{i_\tau}) \geq \p^{\tau'+1}(\x^{\tau'+1}_{i_\tau})$.
\end{enumerate}

Fact I makes the above formulation well-defined, whereas Facts II and III give us the desired implication via the following chain of inequalities:
$$s(t_2) \geq s(\tau) = \p^\tau(\x^\tau_{i_\tau}) \geq \p^{\tau'+1}(\x^{\tau'+1}_{i_\tau}) \geq \alpha_{t_1} s(t_1).$$

The first inequality holds because the spending of the least spender is non-decreasing with time (\Cref{lem:Least_Spender's_Spending_Never_Decreases}), the equality denotes change of notation, and the final two inequalities follow from Facts II and III. The remainder of the proof consists of proving Facts I-III.

\paragraph{Proof of Fact I:}
Suppose, for contradiction, that condition (1) is violated for every $\tau \in [t_1,t_2]$. That is, for every $\tau \in [t_1,t_2]$, the least-spender at $\tau$, namely $i_\tau$, does not experience price-rise during $[t_1,\tau)$.

For each $t \in [t_1,t_2]$, let $\H_t$ denote the set of agents in the hierarchy, and let $M_t$ (the ``marked'' set) denote the set of agents that experience price-rise during $[t_1,t)$. We set $M_{t_1} \coloneqq \{\emptyset\}$, i.e., all agents are ``unmarked'' \emph{before} the price-rise at $t_1$. After the price-rise at $t_1$, we have $M_{t_1+1} = \H_{t_1}$, i.e., we ``mark'' the agents that experience price-rise. Thus, the contradiction hypothesis is equivalent to requiring that for every $\tau \in [t_1,t_2]$, $i_\tau$ is unmarked at time-step $\tau$, i.e., $i_\tau \notin M_\tau$. In particular, we require the agent $i_{t_2}$ to be unmarked at $t_2$.
 
 The update rule for $M_t$ is given by
\[ M_{t+1} =
  \begin{cases}
    M_t       & \quad \text{if $t$ is a Phase 2 step},\\
    M_t \cup \H_t  & \quad \text{if $t$ is a Phase 3 step}.
  \end{cases}
\]

Indeed, no agent experiences price-rise during Phase 2 and therefore the marked set remains unchanged. For a Phase 3 operation at time-step $t$, only the agents in the hierarchy $\H_t$ experience price-rise at $t$, and they are added to the current set $M_t$. Observe that a marked agent never becomes unmarked.

We will call a good \emph{marked} (at $t$) if it is owned by a marked agent (at $t$), i.e., $j' \in [m]$ is marked at $t$ if $j' \in \x^t_k$ for some $k \in M_t$. Notice that at the time-step $t_1 + 1$, the agents in $[n] \setminus \H_{t_1}$ constitute the unmarked set, and no such agent has an alternating path to a marked good, i.e., any good in the set $\textstyle{ \left\{ \cup_{k \in M_{t_1 + 1}} \x_k^{t_1+1} \right\} }$.\footnote{This is because the price-rise at $t_1$ forces the removal of all $\MBB{}$ edges between the agents in $[n] \setminus \H_{t_1}$ and the goods in $\textstyle{ \left\{ \cup_{k \in \H_{t_1}} \x_k^{t_1+1} \right\} }$. Therefore, the set of agents with an alternating path to some good in $\textstyle{ \left\{ \cup_{k \in \H_{t_1}} \x_k^{t_1+1} \right\} }$ is a subset of $\H_t$.}

The key observation underlying our proof is that for any time-step $t \in [t_1+1,t_2)$, if there is no alternating path from an unmarked agent (i.e., an agent in $[n] \setminus M_t$) to a marked good (i.e., a good in the set $\textstyle{ \left\{ \cup_{k \in M_t} \x^t_k \right\} }$), then the same holds for the time-step $t+1$, i.e., there is no alternating path from an agent in $[n] \setminus M_{t+1}$ to a good in the set $\textstyle{ \left\{ \cup_{k \in M_{t+1}} \x^{t+1}_k \right\} }$. Before proving this observation, notice that it gives us the desired contradiction as follows: By the contradiction hypothesis, the agent $i_{t_2}$ is unmarked for all $t \in [t_1,t_2]$. In addition, it can be shown that a marked good never becomes unmarked, i.e., $\textstyle{ \left\{ \cup_{k \in M_t} \x^t_k \right\} \subseteq \left\{ \cup_{k \in M_{t+1}} \x^{t+1}_k \right\} }$. (Indeed, a marked good cannot be allocated to an unmarked agent in Phase 2 because a swap is always along an existing alternating path. Similarly, in Phase 3, any newly created $\MBB{}$ edge must be from some agent in the hierarchy $\H_t$. By the update rule, such an agent is included in the marked set $M_{t+1}$.) In particular, the good $j$ is marked at time-step $t_1+1$, and stays so thereafter. By repeated application of the above observation, we get that there is no alternating path from the agent $i_{t_2}$ to the good $j$ at $t = t_2$. This contradicts the assumption that good $j$ is involved in the price-rise at $t_2$ when $i_{t_2}$ is the least-spender.

We will now prove the aforementioned observation via induction. We have already seen that the base case is true; indeed, since $M_{t_1+1} = \H_{t_1}$, no agent in $[n] \setminus M_{t_1+1}$  has an alternating path to any good in the set $\textstyle{ \left\{ \cup_{k \in M_{t_1+1}} \x^{t_1+1}_k \right\} }$. Let us assume that the observation holds for all $t \in [t_1+1,t']$ and prove that it holds for $t = t' + 1$, where $t'+1 < t_2$. We will use case analysis depending on whether $t'+1$ occurs during Phase 2 or Phase 3.
\begin{itemize}
	\item Suppose that $t'+1$ is a Phase 2 operation. By the contradiction assumption, the least-spender at $t'+1$, namely $i_{t'+1}$, is unmarked at $t'+1$. Furthermore, by the induction hypothesis, there is no alternating path from an unmarked agent to a marked good before the swap operation at $t'+1$. This implies that all agents and goods in the hierarchy at $t'+1$, namely $\H_{t'+1}$, are unmarked. Any swap operation must therefore occur between two unmarked agents. Since the swap operation does not create a new $\MBB{}$ edge, any newly created alternating path must involve the swapped good and the new owner of this good. However, the new owner did not have any alternating path to a marked good before the swap operation at $t'+1$. As a result, a new alternating path from an unmarked agent to a marked good cannot be created.
	 
	\item Now suppose that $t'+1$ is a Phase 3 operation. From the update rule, we know that $M_{t'+2} = M_{t'+1} \cup \H_{t'+1}$. Since a price-rise operation does not change the allocation, we have that $\textstyle{ \left\{ \cup_{k \in M_{t'+2}} \x^{t'+2}_k \right\} = \left\{ \cup_{k \in M_{t'+1}} \x^{t'+1}_k \right\} \bigcup \left\{ \cup_{k \in \H_{t'+1}} \x^{t'+1}_k \right\} }$. Consider any agent $i \in [n] \setminus M_{t'+2}$ that is unmarked at $t'+2$. Clearly, $i$ must be unmarked at $t'+1$. By the induction hypothesis, $i$ does not have an alternating path to any good in the set $\textstyle{ \left\{ \cup_{k \in M_{t'+1}} \x^{t'+1}_k \right\} }$ at $t'+1$. Furthermore, the price-rise at $t'+1$ forces the removal of any $\MBB{}$ edges from $i$ to any good in the set $\textstyle{ \left\{ \cup_{k \in \H_{t'+1}} \x^{t'+1}_k \right\} }$. Finally, any newly created $\MBB{}$ edge due to the price-rise is from some agent in $\H_{t'+1}$, and any such agent is added to the marked set $M_{t'+2}$. Therefore, once again, the set of agents that have an alternating path to some good in the set $\textstyle{ \left\{ \cup_{k \in M_{t'+2}} \x_k^{t'+2} \right\} }$ is a subset of $M_{t'+2}$. This finishes the proof of Fact I.
\end{itemize}

\paragraph{Proof of Fact II:}

Our proof of Fact II uses case analysis depending on whether $\tau' = t_1$ or $\tau' > t_1$. First, suppose that $\tau' = t_1$. In this case, we have that
$$\p^{\tau'+1}(\x^{\tau'+1}_{i_\tau}) = \alpha_{t_1} \p^{\tau'}(\x^{\tau'}_{i_\tau}) \geq \alpha_{t_1} \p^{\tau'}(\x^{\tau'}_{i_{\tau'}}) = \alpha_{t_1} s(t_1),$$
where the first equality holds because $i_{\tau}$ experiences price-rise at $\tau'$ by a factor $\alpha_{\tau'} = \alpha_{t_1}$, the inequality holds because $i_{\tau'}$ is the least-spender at $\tau'$, and the final equality uses alternative notation.

Next, suppose that $\tau' > t_1$. We now have that
$$\p^{\tau'+1}(\x^{\tau'+1}_{i_\tau}) = \alpha_{\tau'} \p^{\tau'} (\x^{\tau'}_{i_{\tau}}) \geq \alpha_{\tau'} \p^{\tau'} (\x^{\tau'}_{i_{\tau'}}) \geq (1+\varepsilon) \p^{\tau'} (\x^{\tau'}_{i_{\tau'}}) \geq (1+\varepsilon) s(t_1 + 1) \geq \alpha_{t_1} s(t_1),$$
where the equality holds because $i_\tau$ experiences price-rise at $\tau'$, the first inequality holds because $i_{\tau'}$ is the least-spender at $\tau'$, the second inequality holds because the price-rise factor is always at least $(1+\varepsilon)$,\footnote{When $\alpha = \alpha_3$, Line~\ref{algline:alpha_3} of Algorithm~\ref{alg:Main} forces $\alpha$ to be a positive integral power of $(1+\varepsilon)$. The same implication holds for the case $\alpha = \alpha_1$ as argued in the proof of \Cref{lem:ALG_Phase3_RunningTimeBound}. Note that we do not need to consider the case $\alpha = \alpha_2$ since it only occurs right before the termination of $\Alg{}$.} the third inequality holds because the spending of the least-spender is non-decreasing with time (\Cref{lem:Least_Spender's_Spending_Never_Decreases}), and the final inequality holds by definition of $\alpha_3$ (see Line~\ref{algline:alpha_3} of Algorithm~\ref{alg:Main}). This proves Fact II.

\paragraph{Proof of Fact III:}
From condition (2), we know that for all $t \in [\tau'+1,\tau)$, the least-spender at $t$, namely $i_t$, is unmarked at $t$. By contrast, the agent $i_\tau$ remains marked throughout $[\tau'+1,\tau)$. As in the proof of Fact I, note that no unmarked agent has an alternating path to a marked agent during $[t_1 + 1, \tau)$. We therefore have that none of the least-spenders $i_{\tau'+1},i_{\tau'+2},\dots,i_{\tau-1}$ have $i_\tau$ in their hierarchy, i.e., $i_\tau \notin \H_{\tau'+1} \cup \H_{\tau'+2} \cup \dots \cup \H_{\tau-1}$. This, in particular, implies that $i_\tau$ does not lose a good via a swap operation during any of these time-steps, which is the only way in which its spending can drop. This proves Fact III, and with it, finishes the proof of \Cref{lem:Bound_on_final_price_of_a_fixed_good}.
\end{proof}

The following result is a straightforward consequence of \Cref{lem:Bound_on_final_price_of_a_fixed_good}.

\begin{corollary}\label{cor:Bound_on_overall_spending}
$\p([m]) \leq m^3 v_{\max}^3$, where $\p$ is the price vector returned by $\Alg{}$ upon termination.
\end{corollary}

\SmallDeltaPO*
\begin{proof}
Suppose, for contradiction, that the allocation $\x$ is Pareto dominated by another integral allocation $\y$ in the instance $\I$. That is, $v_k(\y_k) \geq v_k(\x_k)$ for each agent $k \in [n]$ and $v_i(\y_i) > v_i(\x_i)$ for some agent $i \in [n]$. Integrality of valuations in $\I$ implies that $v_i(\y_i) \geq v_i(\x_i) + 1$.

For any agent $k$, let $\alpha_k$ and $\alpha'_k$ denote maximum bang per buck ratios (with respect to the price vector $\p$) in the instances $\I$ and $\I'$ respectively. Recall that for the $\varepsilon$-rounded version $\I'$, we have $v_{k,j} \leq v'_{k,j} \leq (1+\varepsilon)v_{k,j}$ for each agent $k$ and each good $j$. Thus,
$\alpha_k = \max_{j \in [m]} v_{k,j}/p_j \leq \max_{j \in [m]} v'_{k,j}/p_j = \alpha'_k.$
We therefore have
\begin{alignat*}{2}
 \frac{v_k(\x_k)}{\alpha_k} & \geq \frac{v'_k(\x_k)}{(1+\varepsilon) \alpha_k} & \qquad \text{(since $\I'$ is $\varepsilon$-rounded)} &\\
 & \geq \frac{v'_k(\x_k)}{(1+\varepsilon) \alpha'_k} & \qquad \text{(since $\alpha_k \leq \alpha'_k$)} &\\
 & = \frac{ \p(\x_k) } {1+\varepsilon}, & \text{(via $\MBB$ condition in $\I'$)} &
\end{alignat*}
or equivalently, 
\begin{equation}
\frac{v_k(\x_k)}{\p(\x_k)} \geq \frac{\alpha_k} {1+\varepsilon}.
\label{eqn:Small_delta_PO_temp1}
\end{equation}
 In other words, the allocation $\x$---which is guaranteed to $\fPO{}$ for the instance $\I'$---is \emph{close} to being $\fPO{}$ for the original instance $\I$. The remainder of the proof will show that for a small enough $\varepsilon$, the allocation $\x$ turns out to be $\PO$ for the instance $\I$.

Consider the allocation $\y$. By definition of the maximum bang per buck ratio, we have that $\alpha_k \p(\y_k) \geq v_k(\y_k)$ for each agent $k \in [n]$. Since $\y$ Pareto dominates $\x$ in the instance $\I$, we have $\alpha_k \p(\y_k) \geq v_k(\x_k)$, which, along with \Cref{eqn:Small_delta_PO_temp1}, implies that 
\begin{equation}
\p(\y_k) \geq \frac{\p(\x_k)}{1+\varepsilon}.
\label{eqn:Small_delta_PO_temp2}
\end{equation}
Using a similar reasoning for the agent $i$ (and the observation that $v_i(\y_i) \geq v_i(\x_i) + 1$), we get
\begin{equation}
\p(\y_i) \geq \frac{\p(\x_i)}{1+\varepsilon} + \frac{1}{\alpha_i}.
\label{eqn:Small_delta_PO_temp3}
\end{equation}
The combined spending over all goods can be rewritten as follows:
\begin{alignat*}{2}
 \p([m]) & = \sum_{k \in [n]} \p(\y_k) && \text{(since all goods are allocated under $\y$)}\\
  & = \p(\y_i) + \sum_{k \in [n] \setminus \{i\}} \p(\y_k) &&\\
  & \geq \frac{\p(\x_i)}{1+\varepsilon} + \frac{1}{\alpha_i} + \sum_{k \in [n] \setminus \{i\}} \frac{\p(\x_k)}{1+\varepsilon} && \qquad (\text{from \Cref{eqn:Small_delta_PO_temp2,eqn:Small_delta_PO_temp3}})\\
  & = \frac{\p([m])}{1+\varepsilon} + \frac{1}{\alpha_i}&& \text{(since all goods are allocated under $\x$)}.
\end{alignat*}
This simplifies to
\begin{equation}
\varepsilon \left( \p([m] \alpha_i - 1 \right) \geq 1.
\end{equation}
It is easy to see that $\alpha_i \leq v_{\max}$, since the initial price of each good is at least $1$ (by integrality of valuations), and prices cannot decrease during the execution of $\Alg{}$. Furthermore, from \Cref{cor:Bound_on_overall_spending}, we know that $\p([m]) \leq m^3 v_{\max}^3$. Combining these observations, we get that $\varepsilon \geq \frac{1}{m^3 v_{\max}^4}$, which contradicts the choice of $\varepsilon$.
\end{proof}

\section{Appendix-II}
 \label{sec:Appendix_II}

\subsection{Corner Cases}
\label{subsec:CornerCases}


In \Cref{subsec:Proof_Of_ALG_RunningTime_PowersOfr}, we showed two results---namely, \Cref{lem:ALG_Phase3_RunningTimeBound,lem:ALG_Phase2_NumberOfIdentityChanges}---that provide running time bounds for Phase 2 and Phase 3 by showing that the spending of the least spender increases by some \emph{multiplicative} factor in every polynomial number of steps. These results implicitly assume that the spending of the least spender is nonzero to begin with. In this section, we will show that this assumption holds without loss of generality. 

Our reasoning will depend on whether or not a given fair division instance is a \emph{Hall's violator} \citep{H35representatives}. We will show that if an instance satisfies Hall's condition (i.e, is not a Hall's violator), then the spending of the least spender (in $\Alg$) becomes nonzero in $\O(n^2)$ steps. If, on the other hand, the instance is a Hall's violator, then we can break down its analysis into (a) a smaller instance that satisfies Hall's condition, and (b) a trivial instance.

Consider any fair division instance $\I = \langle [n], [m], \V \rangle$. Write $\G = ([n],[m],E)$ to denote the unweighted bipartite graph between the set of agents and the set of goods such that $E = \{(i,j) \, : \, v_{i,j}>0\}$. The instance $\I$ is said to be a \emph{Hall's violator} if there exists a set of agents $T\subseteq [n]$ that together (positively) value at most $|T|-1$ goods, i.e., the set $T$ violates Hall's condition in the graph $\G$. We call the set $T$ a \emph{maximal} Hall's violator if no strict superset of $T$ is a Hall's violator.

We will first show that an instance can be checked for Hall's condition in $\O(\poly(n,m))$ time. Let $M$ be a maximum (unweighted) matching of $\G$. A given instance $\I$ is a Hall's violator if and only if there exists an agent that is unmatched under $M$. The matching $M$ can be computed in $\O(\poly(n,m))$ time.

\paragraph{Instances that satisfy Hall's condition}
Our first result in this section (\Cref{lem:not_hall_violator}) pertains to instances that satisfy Hall's condition.

\begin{lemma}\label{lem:not_hall_violator}
Let $\I = \langle [n], [m], \V \rangle$ be an input instance to \Alg{} that satisfies power-of-$(1+\varepsilon)$ and Hall's conditions. Then, at each time step after the first $\O(n^2)$ steps, the spending of each agent under \Alg{} is strictly greater than zero.
\end{lemma}
\begin{proof}
 Observe that once the spending of an agent becomes nonzero during the run of \Alg{}, it can never become zero again. This is because for the spending of an agent to drop back to zero, it must lose its last good via a swap operation in Phase 2. However, such an exchange is disallowed by the $\varepsilon$-path-violator condition. Therefore, it suffices to show that the spending of each agent under \Alg{} becomes nonzero after the first $\O(n^2)$ steps.

 Let agent $i$ be the least spender at the end of Phase 1 of \Alg{}. Assume, without loss of generality, that the spending of agent $i$ at the end of Phase 1 is zero (otherwise the lemma follows immediately). We will show that after $\O(n)$ steps, the spending of agent $i$ strictly exceeds zero. The desired running time bound of $\O(n^2)$ will follow from a similar argument for the other agents.
 
 Our proof for the above claim consists of case analysis for whether or not at the end of Phase 1, there exists some agent in the hierarchy $\H_i$ that owns two or more goods. Suppose there exists an agent $k \in \H_i$ that owns two or more goods (if there are multiple such agents, tie-break in favor of agents at a lower level in $\H_i$, and then according to a prespecified lexicographic ordering). Then, $k$ must be an $\varepsilon$-violator, and therefore also an $\varepsilon$-path-violator (along some alternating path $P$). Additionally, by the choice of agent $k$, no agent at a lower level is an $\varepsilon$-path-violator. As a result, $k$ must lose a good under a swap operation in Phase 2 to its predecessor along path $P$, who acquires two goods as a result, and becomes the new $\varepsilon$-path-violator. This series of swaps continues for $\O(n)$ steps, and ends with the least spender receiving a new good.

Next, suppose that each agent in $\H_i$ owns exactly one good. Since the  given instance $\I$ satisfies Hall's condition, and agent $i$ does not yet own any good, there must exist an agent $k \notin \H_i$ that owns two or more goods. We will show that such an agent must get added to the hierarchy in $\O(n)$ steps. Then, by the above argument, in further $\O(n)$ steps, agent $i$ must receive a good that takes its spending strictly above zero.

Since each agent in the hierarchy owns exactly one good, there are no $\varepsilon$-path-violators in $\H_i$, and \Alg{} proceeds directly to Phase 3. Once again, since the least spender does not own any good, its spending cannot change as a result of price-rise, and therefore a new agent gets added to the hierarchy. If this agent has more than one good, then the lemma follows. Otherwise, the price-rise step is repeated. Therefore, after $\O(n)$ such steps, an agent with two or more goods must get added to the hierarchy, as desired.
\end{proof}

\paragraph{Instances that violate Hall's condition} If the given instance $\I = \langle [n], [m], \V \rangle$ is a Hall's violator, then we will first find a maximum (cardinality) matching $M$ in $\G$. Write $A \subset [n]$ to denote the set of agents that get matched in $M$. Note that the instance $\I'= \langle A, [m], \V'\rangle$ satisfies Hall's condition; here $\V'$ is the set of valuations of the agents in $A$. Also, the maximality of $M$ ensures that for every good $j \in [m]$, there exists $i \in A$ such that $v_{i,j} >0$. We will show that the allocation obtained by applying \Alg{} on $\I'$ is $\EFone{}$ and $\fPO{}$ for $\I$. 

Let $\x$ and $\p$ denote the allocation and price vector retuned by \Alg{} on $\I'$. For each of the remaining agents $i \in [n]\setminus A$, set $\x_i = \{\emptyset\}$. By construction, $\x$ satisfies the $\EFone{}$ condition for the agents in $A$. Hence, in order to show that $\x$ is $\EFone{}$ for $\I$, we need to show that this condition also holds for the agents in $[n]\setminus A$. Suppose, for contradiction, that there exists an agent $b \in [n]\setminus A$ that $\EFone{}$ envies $a \in A$. This happens if and only if $a$ is allocated two or more goods valued by $b$, i.e., $ |\x_a \cap \{j \in [m] \, : \, v_{b,j} > 0 \}| \geq 2 $. We will show that this contradicts the fact that $M$ is a maximum matching, and hence prove that $\x$ is an $\EFone{}$ allocation. 

Since $\I'$ satisfies Hall's condition, we have from \Cref{lem:not_hall_violator} that each agent in $A$ has nonzero spending under $(\x,\p)$, i.e., for each agent $i \in A$, we have $|\x_i| \geq 1$. Consider a new matching $M'$ (in $\G$) wherein $ i \in A \setminus \{ a \}$ is matched to a good in $\x_i$, $a$ is matched to some good $j' \in \x_a$, and $b$ is matched to a good from the non-empty set $\{ \x_a \setminus \{j'\} \} \cap \{ j \, : \, v_{b,j} > 0\}$. The size of the matching $M'$ is strictly greater than the size of $M$, which is a contradiction. Hence, $\x$ must be $\EFone{}$ for $\I$. 

Finally, to show that $\x$ is $\fPO$ for $\I$, we can set the endowments of all the agents in $[n]\setminus A$ to be zero, and obtain a market equilibrium $(\x,\p)$ for $\I$. Since $\x$ is the equilibrium outcome of some Fisher market, we have from \Cref{prop:FirstWelfareTheorem} that $\x$ must be $\fPO$ for $\I$.

\subsection{Additional Market Preliminaries}
\label{sec:Market_Appendix}

\paragraph{Eisenberg-Gale program} 

The convex program of \citet{EG59consensus} is known to characterize the equilibria of the Fisher market.

\begin{equation*}
\begin{array}{ll@{}ll}
\text{maximize}  & \sum\limits_{i=1}^n e_i \cdot \ln(u_i) &\\
\text{subject to}& u_i = \sum\limits_{j=1}^m v_{i,j} x_{i,j}  & \quad \forall \, i \in [n],\\
& \sum\limits_{i=1}^n x_{i,j} \leq 1 & \quad \forall \, j \in [m], \text{ and}\\
&x_{i,j} \geq 0 & \quad \forall \, i \in [n] \text{ and } j \in [m].
\end{array}
\end{equation*}

\paragraph{KKT conditions and maximum bang per buck}
We will now describe the KKT conditions for the Eisenberg-Gale program. Let $p_j$ denote the Lagrangian variable corresponding to the constraint $\sum_{i=1}^n x_{i,j} \leq 1$. Then,
\begin{enumerate}
	\item \emph{Dual feasibility}: $p_j \geq 0$ for each $j \in [m]$.
	\item \emph{Complementary slackness}:
	\begin{enumerate}
		\item For each $j \in [m]$, $p_j > 0 \implies \sum_{i=1}^n x_{i,j} = 1.$
		\item For each $i \in [n]$ and $j \in [m]$, $\frac{v_{i,j}}{p_j} \leq \frac{\sum_{j=1}^m v_{i,j} x_{i,j}}{e_i}.$
		\item For each $i \in [n]$ and $j \in [m]$, $x_{i,j} > 0 \implies \frac{v_{i,j}}{p_j} = \frac{\sum_{j=1}^m v_{i,j} x_{i,j}}{e_i}.$
	\end{enumerate}	 
\end{enumerate}

Under the assumptions that each good has an interested buyer (i.e., for each good $j$, $v_{i,j} > 0$ for some buyer $i \in [n]$) and each buyer is interested in some good (i.e., for each buyer $i \in [n]$, $v_{i,j} > 0$ for some good $j \in [m]$), it can be verified that optimal solutions of the Eisenberg-Gale program characterize the market equilibria of the corresponding Fisher market.

\subsection{Second Welfare Theorem for Fisher Markets}
\label{subsec:Second_Welfare_Theorem_for_Fisher_Markets}

\begin{restatable}{theorem}{Second Welfare Theorem for Fisher Markets}
 \label{thm:SecondWelfareTheoremFisherMarkets}
 Let $\I = \langle [n], [m], \V \rangle$ be an instance of the fair division problem, and let $\x$ be a fractionally Pareto efficient $(\fPO{})$ allocation for $\I$. Then, there is a price vector $\p = (p_1,\dots,p_m)$ and an endowment vector $\e = (e_1,\dots,e_n)$ such that $(\x,\p)$ is a market equilibrium for the market instance $\langle [n], [m], \V, \e \rangle$.
\end{restatable}
 \begin{proof}
 Our proof relies on the formulation of a linear program that characterizes the set of all $\fPO{}$ allocations with respect to a given set of utilities. The market equilibrium conditions then follow from linear programming duality and complementary slackness conditions.
 
 Let $\u = (u_1,\dots,u_n)$ denote the utility vector induced by the given $\fPO{}$ allocation $\x$. That is, $u_i = \sum_{j=1}^m v_{i,j} x_{i,j}$ for each $i \in [n]$. Consider the linear program given by \Cref{eqn:fPO_LP_Primal}: 
\begin{equation}
\begin{array}{ll@{}ll}
\text{maximize}  & \sum\limits_{i=1}^n \sum\limits_{j=1}^m v_{i,j} y_{i,j} &\\
\text{subject to}& \sum\limits_{j=1}^m v_{i,j} y_{i,j} \geq u_i  & \quad \forall \, i \in [n],\\
& \sum\limits_{i=1}^n y_{i,j} \leq 1 & \quad \forall \, j \in [m], \text{ and}\\
&y_{i,j} \geq 0 & \quad \forall \, i \in [n] \text{ and } j \in [m].
\end{array}
\label{eqn:fPO_LP_Primal}
\end{equation}
A feasible solution of \Cref{eqn:fPO_LP_Primal} is a fractional allocation $\y$ that provides each buyer $i \in [n]$ with utility at least $u_i$. Notice that the objective in \Cref{eqn:fPO_LP_Primal} is precisely the sum of utilities of all buyers under $\y$. Since the allocation $\x$ is $\fPO{}$, we know that the optimal objective value must be equal to $\sum_{i=1}^n u_i$. Hence, $\x$ is a primal optimal solution.

The dual of the program in \Cref{eqn:fPO_LP_Primal} is given by \Cref{eqn:fPO_LP_Dual} below:
\begin{equation}
\begin{array}{ll@{}ll}
\text{minimize} & \sum\limits_{j=1}^m p_j  - \sum\limits_{i=1}^n u_i d_i &\\
\text{subject to}& p_j - d_i v_{i,j} \geq v_{i,j}  & \quad \forall \, i \in [n] \text{ and } j \in [m],\\
& d_i \geq 0 & \quad \forall \, i \in [n], \text{ and}\\
&p_j \geq 0 & \quad \forall \, j \in [m].
\end{array}
\label{eqn:fPO_LP_Dual}
\end{equation}
The dual variables $d_i$ and $p_j$ correspond to the first and second primal constraints respectively.

Let $(\x,\p^*,\d^*)$ be a tuple of optimal primal and dual solutions given a utility vector $\u$ induced by the $\fPO{}$ allocation $\x$. As discussed above, the optimal primal objective must be $\sum_{i=1}^n u_i$. Therefore, by strong duality,
\begin{equation}
	\sum_{i=1}^n u_i = \sum_{j=1}^m p^*_j  - \sum_{i=1}^n u_i d^*_i.
\label{eqn:temp1}
\end{equation}

Furthermore, dual feasiblity implies that for every $i \in [n]$ and $j \in [m]$,
\begin{equation}
\begin{array}{cccl}
p^*_j - d^*_i v_{i,j} \geq v_{i,j} & \implies & p^*_j & \geq (1+d^*_i) v_{i,j}\\
& \implies & 1+d^*_i & \leq \frac{p^*_j}{v_{i,j}}\\
& \implies & 1+d^*_i & \leq \min_{k} \frac{p^*_k}{v_{i,k}}.
\end{array}
\label{eqn:temp2}
\end{equation}

The final inequality in \Cref{eqn:temp2} in fact holds with an equality (otherwise the dual objective can be improved without violating feasibility). Thus, for each $i \in [n]$,
\begin{equation}
1+d^*_i = \min_{k} \frac{p^*_k}{v_{i,k}}.
\label{eqn:temp3}
\end{equation}

Let $\e = (e_1,\dots,e_n)$ be an endowment vector defined as $e_i = \sum_{j=1}^m p^*_j x_{i,j}$ for every $i \in [n]$. We claim that $(\x,\p^*)$ is a market equilibrium for the market instance $\langle [n], [m], \V, \e \rangle$. The proof of the claim follows from checking the market equilibrium conditions, as below:
\begin{itemize}
	\item \emph{Market clearing}: Recall that $p_j$ is the dual variable corresponding to the primal constraint $\sum_{i=1}^n y_{i,j} \leq 1$. Hence, by complementary slackness, for each $j \in [m]$, either $p^*_j = 0$ or $\sum_{i=1}^n x_{i,j} = 1$, which is precisely the market clearning condition.
	\item \emph{Budget exhaustion}: This follows from the choice of the endowment vector $\e$.
	\item \emph{$\MBB{}$ allocation}: Notice that $y_{i,j}$ is the primal variable corresponding to the dual constraint $p_j - d_i v_{i,j} \geq v_{i,j}$. Therefore, by complementary slackness, we have that
\begin{equation*}
\begin{array}{ccrlr}
x_{i,j} > 0 & \implies & p^*_j - d^*_i v_{i,j} & = v_{i,j} &\\
& \implies & p^*_j & = (1 + d^*_i) v_{i,j} &\\
& \implies & \frac{v_{i,j}}{p^*_j} & = \max_{k} \frac{v_{i,k}}{p^*_k} & (\text{from \Cref{eqn:temp3}}),
\end{array}
\end{equation*}
which gives the $\MBB{}$ allocation condition.
\end{itemize}
We have therefore shown the existence of a price vector $\p^*$ and an endowment vector $\e^*$ such that $(\x,\p)$ is a market equilibrium for the market instance $\langle [n], [m], \V, \e \rangle$, as desired.
\end{proof}

\subsection[Utility Maximization Does Not Imply MBB Condition in Fisher Markets]{Utility Maximization Does Not Imply MBB Condition in Fisher Markets}
\label{subsec:Fisher_counterexample}

Consider a Fisher market instance with two buyers and three goods. The initial endowments are $e_1 = 130$ for buyer 1 and $e_2 = 50$ for buyer 2. The valuations of the buyers are as follows: $v_{1,1} = 100$, $v_{1,2} = 50$, $v_{1,3} = 1$, $v_{2,1} = 1$, $v_{2,2} = 99$, and $v_{2,3} = 100$.

Let $\p = (p_1,p_2,p_3)$ be a price vector, where $p_1 = 70$, $p_2 = 60$, and $p_3 = 50$. The bang per buck ratios are given by $\alpha_{1,1} = \frac{10}{7}$, $\alpha_{1,2} = \frac{5}{6}$, $\alpha_{1,3} = \frac{1}{50}$, $\alpha_{2,1} = \frac{1}{70}$, $\alpha_{2,2} = \frac{99}{60}$, and $\alpha_{2,3} = \frac{100}{50}$. Thus, $\MBB_1 = \{g_1\}$ and $\MBB_2 = \{g_3\}$. Consider an allocation $\x$ given by $\x_1 = \{g_1,g_2\}$ and $\x_2 = \{g_3\}$. Notice that $\x$ does not satisfy the maximum bang per buck condition with respect to $\p$. However, the pair $(\x,\p)$ is utility maximizing under the given budget constraints.

\subsection{An Instance where a 1.44-approximate NSW and \fPO{} Allocation Does Not Exist}
\label{subsec:approxNSW+fPO_non-existance}

This section provides an example of a fair division instance where the Nash social welfare of any fractionally Pareto efficient $(\fPO{})$ allocation is at most $\frac{1}{1.44}$ times that of the Nash optimal allocation.

We define an instance $\I$ with $n=3$ agents and $m = 5$ goods. The set of goods is divided into two \emph{high-valued} goods $\{h_1,h_2\}$ and three \emph{signature} goods $\{g_1,g_2,g_3\}$. Each agent $i \in [n]$ values a high-valued good $h_j$ at $v_{i,j} = c$ (where $c$ is a large constant). Each signature good $g_j$ is valued by the agent with the same index $j \in [n]$ at $v_{j,j} = \frac{1}{3}$, and by every other agent $i \in [n] \setminus \{j\}$ at $v_{i,j} = \frac{1}{3} - \delta$, where $\delta \in \left(0,\frac{3\varepsilon}{2(1+3\varepsilon)} \right)$ is a prespecified constant.

Let $\x$ be any integral $\fPO{}$ allocation. \Cref{claim:fPO-example} below asserts that there must exist an agent $i \in [n]$ such that $\x_i \subseteq \{ g_i \}$. 

\begin{claim}
\label{claim:fPO-example}
	Let $\x$ be any integral $\fPO{}$ allocation with respect to the instance $\I$. Then, there exists an agent $i\in [n]$ such that $\x_i \subseteq \{ g_i \}$.
\end{claim}
\begin{proof}	
Observe that the number of high-valued goods is strictly smaller than the number of agents. Hence, some agent (say agent $1$) must miss out on a high-valued good under the allocation $\x$. If $\x_1 \subseteq \{g_1\}$, then the claim follows. So, let us assume that agent $1$ gets the signature good of some other agent (say agent $2$), i.e., $g_2 \in \x_1$. We will now show that the claim must hold for the agent $2$.

Since $\x$ is $\fPO{}$, we have $g_2 \in \MBB_1$. In addition, from \Cref{thm:SecondWelfareTheoremFisherMarkets}, we know that there exists a price vector $\p$ such that $(\x,\p)$ is a market equilibrium. Hence, the bang per buck ratio of agent $1$ for the good $g_2$ is at least that for any other good. That is,
\begin{equation}
	\alpha_{1,g_2} \geq \alpha_{1,j} \text{ for every } j \in [m].
\label{eqn:approxMMS+fPO_nonexistence_temp1}
\end{equation}
Furthermore, since $v_{1,j} = v_{2,j}$ for every $j \in [m] \setminus \{g_1,g_2\}$, we have that
\begin{equation}
	\alpha_{1,j} = \alpha_{2,j} \text{ for every } j \in [m] \setminus \{g_1,g_2\}.
\label{eqn:approxMMS+fPO_nonexistence_temp2}
\end{equation}

Next, since $v_{2,g_2} > v_{1,g_2}$ by construction, we also have that $\alpha_{2,g_2} > \alpha_{1,g_2}$. Along with \Cref{eqn:approxMMS+fPO_nonexistence_temp1,eqn:approxMMS+fPO_nonexistence_temp2}, this gives
\begin{alignat*}{4}
	\alpha_{2,g_2} & > & \; & \alpha_{1,g_2} & \quad &&&\\
	                 & \geq & \; & \alpha_{1,j} \text{ for every } j \in [m]  & \quad  \text{(\Cref{eqn:approxMMS+fPO_nonexistence_temp1})} &&&\\
                     & = & \; & \alpha_{2,j} \text{ for every } j \in [m] \setminus \{g_1,g_2\} & \quad \text{(\Cref{eqn:approxMMS+fPO_nonexistence_temp2})} &.&&
\end{alignat*}

 Similarly, since $v_{1,g_1} > v_{2,g_1}$, we have that $\alpha_{1,g_1} > \alpha_{2,g_1}$. Combining this with the above implications, we have that $\alpha_{2,g_2} > \alpha_{2,j}$ for every $j \in [m] \setminus \{g_2\}$. In other words, the $\MBB{}$ set of agent $2$ consists only of the good $g_2$. Since $\x$ is an $\fPO{}$ allocation, we have that $\x_2 \subseteq \MBB_2 = \{g_2\}$. Hence, the claim holds for the agent $2$.
\end{proof}

Among all allocations that are $\fPO{}$ for the instance $\I$, let $\x$ denote the one with the largest Nash social welfare. From Claim~\ref{claim:fPO-example}, we know that there exists some agent, say agent 3, such that $\x_3 \subseteq \{g_3\}$. If $\x_3 = \{ \emptyset \}$, then $\NW(\x)=0$. However, since there exists an $\fPO{}$ allocation with nonzero Nash social welfare (namely the welfare maximizing allocation), we must have that $\x_3=\{g_3\}$. Indeed, $\NW(\x) = \left( (c+\frac{1}{3}) \cdot (c+\frac{1}{3}) \cdot \frac{1}{3} \right)^{\frac{1}{3}}$ for the allocation $\x = (\{h_1,g_1\},\{h_2,g_2\},\{g_3\})$.
	
	It is easy to check that the allocation $\y=(\{h_1\},\{h_2\},\{g_1,g_2,g_3\})$ is the Nash optimal allocation (without the $\fPO{}$ constraint), and $\NW(\y)= \left( c^2 \cdot (1 - 2\delta) \right)^{\frac{1}{3}} \approx c^{2/3}$ for small $\delta$. Therefore, 
	$$\frac{\NW(\y)}{\NW(\x)}\geq \left( \frac{c^2}{\frac{1}{3}(c+\frac{1}{3})^2} \right)^\frac{1}{3} = 3^{\frac{1}{3}} \left( \frac{3c}{3c+1} \right)^\frac{2}{3} \geq 1.44 \ \text{ for large } c.$$

\paragraph{An allocation that is $\PO{}$ but not $\fPO{}$} Notice that the allocation $\y$ in the above example is Nash optimal, and hence, by the result of \citet{CKM+16unreasonable}, is Pareto efficient ($\PO$). However, there is no agent that receives only a subset of its signature goods under $\y$. Therefore, from Claim~\ref{claim:fPO-example}, $\y$ cannot be $\fPO$.

\subsection{\EFone{} Allocations can be Highly Inefficient}
\label{subsec:EF1-inefficient}

	Define an instance $\langle [n],[m], \V \rangle$ with $m = n$ such that $v_{i,j} = 1$ if $i = j$ and $0$ otherwise. Consider an allocation $\x$ such that $\x_i=\{g_{i+1}\}$ for all $i\in[n-1]$ and $\x_n=\{g_1\}$. Clearly, $\x$ is an $\EFone$ allocation. However, $\x$ is highly inefficient since each agent gets a valuation of zero.

\subsection{Every Rounding of Spending Restricted Outcome Violates \EFone}
\label{subsec:spending-restricted}

	In this section, we provide an example of a fair division instance where every rounding of the spending restricted equilibrium defined by \citet{CG15approximating} violates $\EFone$ condition. Our example also serves as a counterexample for the rounding of CEEI outcomes as well.
	
	Consider an instance with $n=5$ agents and $m=7$ goods. Let $\{g_1,g_2,\dots,g_7\}$ denote the set of goods, and let $v_1,v_2,\dots,v_5$ denote the valuation functions. 
	
	\begin{center}
		\begin{tabular}{| l | l | l | l | l | l | l | l |}
			\hline
			&$g_1$&$g_2$&$g_3$&$g_4$&$g_5$&$g_6$& $g_7$ \\ \hline
			Agent~$1$ & $3/4$   & 0   & 0   & $3/4$   & 0   & 0   & 0 \\ \hline
			Agent~$2$ & 0   & $3/4$   & 0   & $3/4$   & 0   & 0   & 0 \\ \hline
			Agent~$3$ & 0   & 0   & $3/4$   & $3/4$   & 0   & 0   & 0 \\ \hline
			Agent~$4$ & $0.7$ & $0.7$ & $0.7$ & $0.7$   & $2/3$   & 0   & $2/3$ \\ \hline
			Agent~$5$ & $0.7$ & $0.7$ & $0.7$ & $0.7$   & 0   & $2/3$   & $2/3$ \\ 
			\hline
		\end{tabular}
	\end{center}
	
	The unique CEEI price vector is given by $\p=\left(\frac{3}{4},\frac{3}{4},\frac{3}{4},\frac{3}{4},\frac{2}{3},\frac{2}{3},\frac{2}{3}\right)$, and the unique CEEI fractional allocation is $\y= \left( \{g_1,\frac{g_4}{3}\},\{g_2,\frac{g_4}{3}\},\{g_3,\frac{g_4}{3}\},\{g_5,\frac{g_7}{2}\},\{g_6,\frac{g_7}{2}\} \right)$.
	Since the price of each good is strictly less than $1$, $\y$ is also the unique spending restricted outcome. Let $\x$ be any rounding of the fractional allocation $\y$. Then, there exists an agent $i\in \{1,2,3\}$ and an agent $k\in\{4,5\}$ such that $\x_i=\{g_i,g_4\}$ and $\x_k=\{g_{k+1}\}$. Hence $v_k(\x_k)=\frac{2}{3}< 0.7 = v_k(\x_i\setminus\{j\})$ for any $j\in\x_i$, i.e., the allocation $\x$ is not $\EFone$.

\subsection{Approximate NSW may not be \EFone{} or \PO{}}
\label{subsec:ApproxNash_NotEF1_NotPO}

	Consider an instance $\I = \langle [n],[m],\V \rangle$ with $m=2n$. Each good $j \in [2n-2]$ is valued at $v_{i,j}=2^{n-1}$ by each agent $i \in [n]$. In addition, we have $v_{i,(2n-1)} = v_{i,2n} = 0$ for each $i \in [n-1]$, and $v_{n,(2n-1)}=1$ and $v_{n,2n}=2^n - 1$.
	
	The allocation $\x = \left( \{1,2\},\{3,4\},\dots,\{(2n-1),2n\} \right)$ is Nash optimal, and $\NSW(\x)= 2^n$. The allocation $\y = \left( \{1,2,2n\},\{3,4\},\{5,6\},\dots,\{(2n-1)\} \right)$ is a $2$-approximation to Nash social welfare since $\NSW(\y)=2^{n-1}$. However, the allocation $\y$ is not \EFone{}, since agent $n$ envies every other agent by more than up to one good. Also, $\y$ is not $\PO$ since the allocation $\x$ Pareto dominates $\y$.

\end{document}